\documentclass[sigconf]{acmart}

\usepackage{booktabs}
\usepackage[font=scriptsize,subrefformat=parens]{subcaption}
\usepackage{algorithm}
\usepackage{algpseudocode}
\usepackage{arydshln}
\usepackage{bm}
\usepackage{enumitem}
\usepackage{tikz}
\usetikzlibrary{calc}
\usetikzlibrary{positioning}
\usetikzlibrary{shapes.geometric}
\usetikzlibrary{patterns}
\tikzset{
vertex/.style={circle,draw,black,align=center,inner sep=0cm, minimum size=0.6cm,fill=white,anchor=center},
vertex-s/.style={vertex,minimum size=0.2cm},
dest/.style={circle,draw,black,densely dashed, minimum size=0.7cm},
agent/.style={vertex,fill={rgb:red,1;green,1;blue,1},text=white},
agenthigh/.style={agent,fill={rgb:red,0;green,1;blue,3}},
agentlow/.style={agent,fill={rgb:red,3;green,1;blue,0}},
agentchange/.style={agent,pattern=north east lines,pattern color={rgb:red,0;green,1;blue,3},text=black},
agent0/.style={agenthigh},
agent1/.style={agent,fill={rgb:green,1;blue,1}},
agent2/.style={agent,fill={rgb:red,1;yellow,2;blue,0},text=black},
agent3/.style={agentlow},
agent-d/.style={agent,line width=0.5mm,fill=white,text=black},
agent0-d/.style={agent-d,draw={rgb:red,0;green,1;blue,3}},
agent1-d/.style={agent-d,draw={rgb:green,1;blue,1}},
agent2-d/.style={agent-d,draw={rgb:red,1;yellow,2;blue,0}},
agent3-d/.style={agent-d,draw={rgb:red,3;green,1;blue,0}},
line/.style={black},
move/.style={line,very thick,->},
move-normal/.style={move,color={rgb:red,1;green,1;blue,1}},
move-high/.style={move,color={rgb:red,0;green,1;blue,3}},
move-low/.style={move,color={rgb:red,3;green,1;blue,0}},
move-ideal/.style={move},
pi/.style={line,very thick,->},
bt/.style={line,double,->},
}
\usepackage{macro}
\usepackage{flushend}

\usepackage{fancyhdr}
\renewcommand\footnotetextcopyrightpermission[1]{}

\begin{document}

\title{winPIBT: Extended Prioritized Algorithm \\for Iterative Multi-agent Path Finding}  

\author{Keisuke Okumura}
\affiliation{%
 \institution{Tokyo Institute of Technology}
}
\email{okumura.k@coord.c.titech.ac.jp}

\author{Yasumasa Tamura}
\affiliation{%
 \institution{Tokyo Institute of Technology}
}
\email{tamura@c.titech.ac.jp}

\author{Xavier D\'{e}fago}
\affiliation{%
 \institution{Tokyo Institute of Technology}
}
\email{defago@c.titech.ac.jp}

\begin{abstract}  
 The problem of Multi-agent Path Finding (MAPF) consists in providing agents with efficient paths while preventing collisions.
 Numerous solvers have been developed so far since MAPF is critical for practical applications such as automated warehouses.
 The recently-proposed Priority Inheritance with Backtracking (PIBT) is a promising decoupled method that solves MAPF iteratively with flexible priorities.
 The method is aimed to be decentralized and has a very low computational cost, but it is shortsighted in the sense that it plans only one step ahead, thus occasionally resulting in inefficient plannings.
 This work proposes a generalization of PIBT, called windowed PIBT (winPIBT), that introduces a configurable time window.
 winPIBT allows agents to plan paths anticipating multiple steps ahead.
 We prove that, similarly to PIBT, all agents reach their own destinations in finite time as long as the environment is a graph with adequate properties, e.g., biconnected.
 Experimental results over various scenarios confirm that winPIBT mitigates livelock situations occurring in PIBT, and usually plans more efficient paths given an adequate window size.
\end{abstract}


\maketitle

\pagestyle{fancy}
\fancyhf{}
\fancyhead[LO,LE]{winPIBT for Iterative Multi-agent Path Finding}
\fancyhead[RO,RE]{K. Okumura, Y. Tamura, X. D\'efago}
\cfoot{\vspace{2mm}\thepage\ of \ref{TotPages}}


\section{Introduction}
Multi-agent Path Finding (MAPF) is a problem that makes multiple agents move to their destinations without collisions.
MAPF is now receiving a lot of attention due to its high practicality, e.g., traffic control~\cite{dresner2008multiagent}, automated warehouse~\cite{wurman2008coordinating}, or airport surface operation~\cite{morris2016planning}, etc.
The efficiency of planned paths is usually evaluated through the sum of travel time.
Since search space grows exponentially with the number of agents, the challenge is obtaining relatively efficient paths with acceptable computational time.

Considering realistic scenarios, MAPF must be solved iteratively and in real-time since many target applications actually require agents to execute streams of tasks;
MAPF variants tackle this issue, e.g., \textit{lifelong} MAPF~\cite{ma2017lifelong}, \textit{online} MAPF~\cite{vsvancara2019online}, or, \textit{iterative} MAPF~\cite{okumura2019priority}.
In such situations, decoupled approaches, more specifically, approaches based on prioritized planning~\cite{erdmann1987multiple,silver2005cooperative}, are attractive since they can reduce computational cost.
Moreover, decoupled approaches are relatively realistic to decentralized fashion, i.e., each agent determines its own path while negotiating with others.
Thus, they have the potential to receive benefits of decentralized systems such as scalability and concurrency.

\begin{figure}[t]
 \centering
 \begin{tikzpicture}
  \node[vertex](v0) at (0,0) {};
  \node[vertex, right=0.35 of v0](v1) {};
  \node[vertex, right=0.35 of v1](v2) {};
  \node[vertex, right=0.35 of v2](v3) {};
  \node[agenthigh, right=0.35 of v3, label=right:{\scriptsize {\tiny priority:} \textbf{high}}](v4) {$a_{1}$};
  \node[vertex, above=0.35 of v0](v5) {};
  \node[vertex, above=0.35 of v2](v6) {};
  \node[vertex, above=0.35 of v5](v7) {};
  \node[agentlow, right=0.35 of v7, label=above:{\scriptsize {\tiny priority:} low}](v8) {$a_{2}$};
  \node[vertex, above=0.35 of v6](v9) {};
  \foreach \u / \v in {v0/v1,v1/v2,v2/v3,v3/v4,v0/v5,v2/v6,v5/v7,v6/v9,v7/v8,v8/v9}
  \draw[line] (\u) -- (\v);
  \coordinate[below left=0.2cm of v4]    (a0-0);
  \coordinate[below right=0.2cm of v2]   (a0-1);
  \coordinate[above right=0.2cm of v9]   (a0-2);
  \coordinate[above right=0.2cm of v8]   (a0-3);
  \foreach \u / \v in {a0-0/a0-1,a0-1/a0-2}
  \draw[move-high,->](a0-0)--(a0-1)--(a0-2)--(a0-3);
  \coordinate[below right=0.02cm of v8]  (a1-0);
  \coordinate[below left =0.02cm of v9]  (a1-1);
  \coordinate[above left =0.02cm of v6]  (a1-2);
  \coordinate[left       =0.12cm of a1-2](a1-3);
  \coordinate[below left =0.18cm of a1-1](a1-4);
  \coordinate[below right=0.2cm of v7]   (a1-5);
  \coordinate[above right=0.2cm of v0]   (a1-6);
  \coordinate[above left=0.2cm of v4]    (a1-7);
  \draw[move-low,->](a1-0)--(a1-1)--(a1-2)--(a1-3)--(a1-4)--(a1-5)--(a1-6)--(a1-7);
  \coordinate[below=0.2cm of v8]         (ai-0);
  \coordinate[below right=0.42cm of v7]  (ai-1);
  \coordinate[above right=0.42cm of v0]  (ai-2);
  \coordinate[above=0.2cm of v4]         (ai-3);
  \draw[move-ideal,->](ai-0)--(ai-1)--(ai-2)--(ai-3);
  \coordinate[above right=1.5cm of v2](ai-ls);
  \coordinate[right=0.4cm of ai-ls,label=right:{\scriptsize ideal path of $a_{2}$}](ai-lg);
  \draw[move-ideal,->](ai-ls)--(ai-lg);
  \coordinate[above=0.4cm of ai-ls](a1-ls);
  \coordinate[right=0.4cm of a1-ls,label=right:{\scriptsize path of $a_{2}$ (PIBT)}](a1-lg);
  \draw[move-low,->](a1-ls)--(a1-lg);
  \coordinate[above=0.4cm of a1-ls](a0-ls);
  \coordinate[right=0.4cm of a0-ls,label=right:{\scriptsize path of $a_{1}$}](a0-lg);
  \draw[move-high,->](a0-ls)--(a0-lg);
 \end{tikzpicture}
 \caption{Motivating example. $a_{1}$ and $a_{2}$ swap their places.}
 \label{fig:motivating-example}
\end{figure}

Priority Inheritance with Backtracking (PIBT)~\cite{okumura2019priority}, a decoupled method proposed recently, solves iterative MAPF by relying on prioritized planning with a unit-length time window, i.e., it determines only the next locations of agents.
With flexible priorities, PIBT ensures \textit{reachability}, i.e., all agents reach their own destinations in finite time, provided that the environment is a graph with adequate properties, e.g., biconnected.
Unfortunately, the efficiency of the paths planned by PIBT is underwhelming as a result of locality.
This is illustrated in Fig.~\ref{fig:motivating-example} which depicts two actual paths (the red and blue arrows) that PIBT plans when an agent $a_{1}$ has higher priority than an agent $a_{2}$.
In contrast, the black arrow depicts an ideal path for $a_{2}$.
Obviously, the agent with lower priority ($a_{2}$) takes unnecessary steps.
This comes as a result of the shortsightedness of PIBT, i.e., PIBT plans paths anticipating only a single step ahead.
Extending the time window is hence expected to improve overall path efficiency thanks to better anticipation.

In this study, we propose a generalized algorithm of PIBT with respect to the time window, called Windowed PIBT (\winpibt).
\winpibt allows agents to plan paths anticipating multiple steps ahead.
Approximately, for an agent $a_i$, \winpibt works as follows.
At first, compute the shortest path while avoiding interference with other paths.
Then, try to secure time-node pairs sequentially along that path (request).
If the requesting node is the last node assigned to some agent $a_j$, then keep trying to let $a_j$ plan its path one step ahead and move away from the node by providing the priority of $a_i$ until there are no such agents.
The special case of \winpibt with a unit-length window is hence similar to PIBT.

Our main contributions are two-folds:
1)~We propose an algorithm \winpibt inheriting the features of PIBT, and prove the reachability in equivalent conditions to PIBT except for the upper bound on time steps.
To achieve this, we introduce a safe condition for paths with different lengths, called \emph{disentangled} condition.
2)~We demonstrate both the effectiveness and the limitation of \winpibt with fixed windows through simulations in various environments.
The results indicate the potential for more adaptive versions.

The paper organization is as follows.
Section~\ref{sec:relatedworks} reviews the existing MAPF algorithms.
Section~\ref{sec:preliminary} defines the terminology and the problem of iterative MAPF, and reviews the PIBT algorithm.
We describe the \safe condition here.
Section~\ref{sec:algo} presents the \winpibt algorithm and its characteristics.
Section~\ref{sec:evaluation} presents empirical results of the proposal in various situations.
Section~\ref{sec:conclusion} concludes the paper and discusses future work.

\section{Related Works}
\label{sec:relatedworks}
We later review PIBT~\cite{okumura2019priority} in detail in section~\ref{subsec:pibt}.

Numerous optimal MAPF algorithms are proposed so far, e.g., search-based optimal solvers~\cite{felner2017search}, however, finding an optimal solution is NP-hard~\cite{yu2013structure}.
Thus, developing sub-optimal solvers is important.
There are complete sub-optimal solvers, e.g.,
BIBOX~\cite{surynek2009novel} for biconnected graphs,
TASS~\cite{khorshid2011polynomial} and multiphase planning method~\cite{peasgood2008complete} for trees.
Push and Swap/Rotate~\cite{luna2011push,de2013push} relies on two types of macro operations;
move an agent towards its goal (push), or, swap the location of two agents (swap).
Push and Swap has several variants, e.g., with simultaneous movements~\cite{sajid2012multi}, or,
with decentralized implementation~\cite{wiktor2014decentralized,zhang2016discof}.
Priority inheritance in (win)PIBT can be seen as ``push'', but note that there is no ``swap'' in (win)PIBT.

Prioritized planning~\cite{erdmann1987multiple} is incomplete but computationally cheap.
The well-known algorithm of prioritized planning for MAPF is Hierarchical Cooperative \astar (\hca)~\cite{silver2005cooperative}, which sequentially plans paths in order of priorities of agents while avoiding conflicts with previously planned paths.
This class of approaches is scalable for the number of agents, and is often used as parts of MAPF solvers~\cite{wang2011mapp,vcap2015prioritized}.
Moreover, prioritized planning are designed to be decentralized, i.e., each agent determines its own path while negotiating with others~\cite{velagapudi2010decentralized,vcap2015prioritized}.
Windowed \hca (\whca)~\cite{silver2005cooperative} is a variant of \hca, which uses a limited lookahead window.
\whca motivates \winpibt since the longer window causes better results in path efficiency and PIBT partly relies on \whca where the window is a unit-length.
Conflict Oriented \whca~\cite{bnaya2014conflict} is an extension of \whca by focusing on the coordination around conflicts, which \winpibt is also focusing on.
Since a priority ordering is crucial, how to adjust priority orders has been studied~\cite{azarm1997conflict,bennewitz2002finding,van2005prioritized,bnaya2014conflict,ma2019searching}.
Similarly to PIBT, \winpibt gives agents their priorities dynamically online so these studies are not closely relevant, however, we say it is an interesting direction to combine these insights into our proposal, especially in initial priorities.
A recent theoretical analysis of prioritized planning~\cite{ma2019searching} identifies instances that fail for any order of static priorities, which motivates planning with \emph{dynamic} priorities, such as taken here.

There are variants of classical MAPF.
Online MAPF~\cite{vsvancara2019online} addresses a dynamic group of agents, i.e., agents newly appear, or, agents disappear when they reach their goals.
Lifelong MAPF~\cite{ma2017lifelong}, defined as the multi-agent pickup and delivery (MAPD) problem, is setting for conveying packages in an automated warehouse.
In MAPD, the system issues goals, namely, pickup and delivery locations, dynamically to agents.
Iterative MAPF~\cite{okumura2019priority} is an abstract model to address the behavior of multiple moving agents, which consists of solving route planning and task allocation.
This model can cover both classical MAPF and MAPD.
We use iterative MAPF to describe our algorithm.

\section{Preliminary}
\label{sec:preliminary}
We now define the terminology, review the PIBT algorithm and introduce \safe condition of paths.

\subsection{Problem Definition}
We first define an abstract model, iterative MAPF.
Then, we destinate two concrete instances, namely, classical MAPF and \naive iterative MAPF.
Both instances only focus on route planning, and task allocation is regarded as input.

The system consists of a set of agents, $A = \{a_{1}, \ldots, a_{n} \}$,
and an environment given as a (possibly directed) graph $G = (V, E)$, where agents occupy nodes in $V$ and move along edges in $E$.
$G$ is assumed to be
1)~\emph{simple}, i.e., devoid of loops and multiple edges,
and
2)~\emph{strongly-connected}, i.e., every node is reachable from every other node.
These requirements are met by simple undirected graphs.
Let $v_{i}(t)$ denote the node occupied by agent $a_{i}$ at discrete time~$t \in \mathbb{N}$.
The initial node $v_{i}(0)$ is given as input.
At each step, an agent $a_{i}$ can either move to an adjacent vertex or stay at the current vertex.
Agents must avoid
1)~\textit{vertex conflict}: $v_{i}(t) \neq v_{j}(t)$,
and
2)~\textit{swap conflict} with others: $v_{i}(t) \neq v_{j}(t+1) \lor v_{j}(t+1) \neq v_{j}(t)$.
Rotations (\emph{cycle conflict}) are not prohibited, i.e.,
$v_{i}(t+1)=v_{j}(t) \land v_{j}(t+1)=v_{k}(t) \land \cdots \land v_{l}(t+1)=v_{i}(t)$
is possible.

Consider a stream of tasks $\Gamma = \{ \tau_{1}, \tau_2, \dots \}$.
Each task is defined as a finite set of goals $\tau_{j} = \{ g_{1}, g_{2}, \dots, g_{m} \}$ where $g_{k} \in V$, possibly with a partial order on $g_k$.
An agent is \emph{free} when it has no assigned task.
Only a free agent can be assigned a task $\tau_{j}$.
When $\tau_{j}$ is assigned to $a_{i}$, $a_{i}$ starts visiting goals in $\tau_{j}$.
$\tau_{j}$ is completed when $a_{i}$ reaches the final goal in $\tau_{j}$ after having visited all other goals, then $a_{i}$ is free again.

The solution includes two parts:
1)~route planning: plan paths for all agents without collisions,
2)~task allocation: allocate a subset of $\Gamma$ to each agent,
such that all tasks are completed in finite time.
The objective function should be determined by concrete instances of iterative MAPF, as shown immediately after.

\subsubsection{Classical MAPF}
A singleton task $\{ g_i \}$ is assigned to each agent $a_{i}$ beforehand, where $g_i$ is a goal for $a_i$.
Since classical MAPF usually requires the solution to ensure that all agents are at their goals simultaneously, a new task $\{ g_{i} \}$ is assigned to $a_{i}$ when $a_{i}$ leaves $g_{i}$.
There are two commonly used objective functions: sum of costs (SOC) and makespan.
SOC is sum of timesteps when \textit{each} agent reaches its given goal and never moves from it.
The makespan is the timestep when \textit{all} agents reach their given goals.

\subsubsection{\Naive Iterative MAPF}
This setting gives a new singleton task, i.e., a new goal, immediately to agents who arrive at their current goals.
We here modify the termination a little to avoid the sensitive effect of the above-defined termination on the performance.
Given a certain integer number $K$, the problem is regarded as solved when tasks issued from 1st to $K$-th are all completed.
The rationale is to analyze the results in operation.
Similarly to classical MAPF, there are two objective functions:
average service time, which is defined as the time interval from task generation
to its completion, or,
makespan, which is the timestep corresponding to the termination.

\subsection{PIBT}
\label{subsec:pibt}
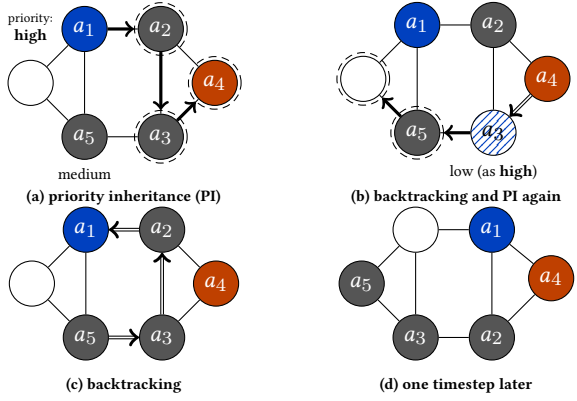
\begin{figure}
 \centering
 \begin{tabular}{cc}
  \begin{minipage}[t]{0.45\hsize}
   \centering
    \begin{tikzpicture}
     \node[vertex](v0) at (0,0) {};
     \node[agent,below right=0.4cm of v0,label=below:{\scriptsize medium}](v1) {$a_{5}$};
     \node[agenthigh,above right=0.4cm of v0](v2) {$a_{1}$};
     \node[agent,right=0.4cm of v1](v3) {$a_{3}$};
     \node[agent,right=0.4cm of v2](v4) {$a_{2}$};
     \node[agentlow,above right=0.4cm of v3](v5) {$a_{4}$};
     \foreach \u / \v in {v0/v1,v0/v2,v1/v3,v2/v4,v3/v5,v4/v5,v1/v2,v3/v4}
     \draw[line] (\u) -- (\v);
     \node[above=0.1cm of v0]() {\scriptsize $\textbf{high}$};
     \node[above=0.35cm of v0](){\scriptsize {\tiny priority:}};
     \node[dest](d0) at (v4) {};
     \node[dest](d1) at (v3) {};
     \node[dest](d2) at (v5) {};
     \draw[pi] (v2)--(d0);
     \draw[pi] (v4)--(d1);
     \draw[pi] (v3)--(d2);
    \end{tikzpicture}
   \vspace{-0.2cm}
   \subcaption{priority inheritance (PI)}
   \label{fig:pibt:pi}
  \end{minipage}
  &
  \begin{minipage}[t]{0.45\hsize}
   \centering
    \begin{tikzpicture}
     \node[vertex](v0) at (0,0) {};
     \node[agent,below right=0.4cm of v0](v1) {$a_{5}$};
     \node[agenthigh,above right=0.4cm of v0](v2) {$a_{1}$};
     \node[agentchange,right=0.4cm of v1,label=below:{\scriptsize low (as \textbf{high})}](v3) {$a_{3}$};
     \node[agent,right=0.4cm of v2](v4) {$a_{2}$};
     \node[agentlow,above right=0.4cm of v3](v5) {$a_{4}$};
     \foreach \u / \v in {v0/v1,v0/v2,v1/v3,v2/v4,v3/v5,v4/v5,v1/v2,v3/v4}
     \draw[line] (\u) -- (\v);
     \node[dest](d2) at (v1) {};
     \node[dest](d4) at (v0) {};
     \draw[pi] (v3)--(d2);
     \draw[pi] (v1)--(d4);
     \draw[bt] (v5)--(v3);
    \end{tikzpicture}
   \vspace{-0.2cm}
   \subcaption{backtracking and PI again}
   \label{fig:pibt:btpi}
  \end{minipage}
  \\
  \begin{minipage}[t]{0.48\hsize}
   \centering
   \begin{tikzpicture}
    \node[vertex](v0) at (0,0) {};
    \node[agent,below right=0.4cm of v0](v1) {$a_{5}$};
    \node[agenthigh,above right=0.4cm of v0](v2) {$a_{1}$};
    \node[agent,right=0.4cm of v1](v3) {$a_{3}$};
    \node[agent,right=0.4cm of v2](v4) {$a_{2}$};
    \node[agentlow,above right=0.4cm of v3](v5) {$a_{4}$};
    \foreach \u / \v in {v0/v1,v0/v2,v1/v3,v2/v4,v3/v5,v4/v5,v1/v2,v3/v4}
    \draw[line] (\u) -- (\v);
    \draw[bt] (v1)--(v3);
    \draw[bt] (v3)--(v4);
    \draw[bt] (v4)--(v2);
   \end{tikzpicture}
   \subcaption{backtracking}
   \label{fig:pibt:bt}
  \end{minipage}
  &
  \begin{minipage}[t]{0.48\hsize}
   \centering
   \begin{tikzpicture}
    \node[agent](v0) at (0,0) {$a_{5}$};
    \node[agent,below right=0.4cm of v0](v1) {$a_{3}$};
    \node[vertex,above right=0.4cm of v0](v2) {};
    \node[agent,right=0.4cm of v1](v3) {$a_{2}$};
    \node[agenthigh,right=0.4cm of v2](v4) {$a_{1}$};
    \node[agentlow,above right=0.5cm of v3](v5) {$a_{4}$};
    \foreach \u / \v in {v0/v1,v0/v2,v1/v3,v2/v4,v3/v5,v4/v5,v1/v2,v3/v4}
    \draw[line] (\u) -- (\v);
   \end{tikzpicture}
   \subcaption{one timestep later}
   \label{fig:pibt:result}
  \end{minipage}
 \end{tabular}
 \caption{Example of PIBT.
 Desired nodes for each agent are depicted by dashed circles.
 Flows of priority inheritance and backtracking are drawn as single-line and doubled-line arrows, respectively.
 As a result of consecutive priority inheritance ($a_{1} \rightarrow a_{2} \rightarrow a_{3}$), $a_{4}$ has no escape nodes (\ref{fig:pibt:pi}).
 To avoid collisions, $a_{1}$, $a_2$ and $a_{3}$ have to wait for backtracking before moving.
 In (\ref{fig:pibt:btpi}), $a_{4}$ sends invalid as backtracking message to $a_{3}$ then $a_{3}$ changes its target node ($a_{5}$).
 In this time, $a_{5}$ can successfully move and sends as valid to $a_{3}$.
 Similarly, $a_{2}$ and $a_{1}$ receive backtracking as valid (\ref{fig:pibt:bt}) and then they start moving (\ref{fig:pibt:result}).
 }
 \label{fig:pibt}
\end{figure}

PIBT~\cite{okumura2019priority} gives fundamental collision-free movements of agents to solve iterative MAPF.
PIBT relies
1)~on \whca~\cite{silver2005cooperative} where the window size is a unit-length,
and
2)~on priority inheritance~\cite{sha1990priority} to deal with \textit{priority inversion} akin to the problem in real-time systems.
At each timestep, unique priorities are assigned to agents.
In order of decreasing priorities, each agent plans its next location while avoiding collisions with higher-priority agents.
When a low-priority agent~$X$ impedes the movement of a higher-priority agent~$Y$, agent~$X$ temporarily inherits the higher-priority of agent~$Y$.
Priority inheritance is executed in combination with \textit{backtracking} to prevent agents being stuck.
The backtracking has two outcomes: valid or invalid.
Invalid occurs when an agent inheriting the priority is stuck, forcing the higher-priority agent to replan its path.
Fig.~\ref{fig:pibt} shows an example of PIBT.
In the sense that PIBT changes priorities to agents dynamically online, PIBT is different from classical prioritized approaches.

The foundation of PIBT is the lemma below, which is also important to \winpibt.
\begin{lemma}
 Let $a_{1}$ denote the agent with highest priority at timestep $t$ and $v_{1}^{\ast}$ an arbitrary neighbor node of $v_{1}(t)$.
 If there exists a simple cycle $\mathbf{C} = (v_{1}(t), v_{1}^{\ast}, \dots)$ and $|\mathbf{C}| \geq 3$, PIBT makes $a_{1}$ move to $v_{1}^{\ast}$ in the next timestep.
 \label{lemma:pibt-local-movement}
\end{lemma}

Another key component is dynamic priorities, where the priority of an agent increments gradually until it drops upon reaching its goal.
By combining these techniques, PIBT ensures the following theorem.
\begin{definition}
 $G$ is \emph{\graphcond} if $G$ has a simple cycle $\bm{C}$ for all pairs of adjacent nodes and $|\bm{C}| \geq 3$.
 \label{def:pibt-cond}
\end{definition}
\begin{theorem}
 If $G$ is \graphcond, PIBT lets all agents reach their own destination within $\text{diam}(G)|A|$ timesteps after the destinations are given.
 \label{theorem:pibt}
\end{theorem}
Examples of \graphcond graphs are undirected biconnected or directed rings.
Note that the above theorem does not say \textit{complete} for classical MAPF, i.e., it does not ensure that all agents are on their goals \emph{simultaneously}.

\subsection{\Safe Condition}
Assume two paths for agents $a_i, a_j$ with different lengths, and
let the corresponding last timesteps of those two paths be $t_i$ and $t_j$ such that $t_i < t_j$.
Assume that no agents collide until $t_i$.
Unless agents vanish after they reach their goals, $a_i$ has to plan its extra path by $t_j$ since two agents potentially collide at some timestep $t$,  $t_i < t \leq t_j$.
However, $a_i$ does not need to compute the extra path immediately if $a_j$ does not use the last node of paths for $a_i$.
This is because the shorter path can be extended so as not to collide with the longer path, i.e. by staying at the last node, meaning that $a_i$ can compute its extra path on demand.
We now define these concepts clearly.

We define a sequence of nodes $\pi_{i}$ as a determined path of an agent $a_{i}$.
Initially, $\pi_{i}$ only contains $v_{i}(0)$.
The manipulation to $\pi_{i}$ only allows to append the latest assigned node.
We use $\ell_{i}$ as the timestep which corresponds to the latest added node to $\pi_{i}$.
Note that $\pi_{i} = \left( v_{i}(0),\dots,v_{i}(\ell_{i}) \right)$ and $\ell_{i} = |\pi_{i}| - 1$ from those definition.
The list of paths of all agents $A$ is denoted by $\bm{\pi}$.
\begin{definition}
 Given two paths $\pi_{i}, \pi_{j}$ and assume that $\ell_{i} \leq \ell_{j}$.
 $\pi_{i}$ and $\pi_{j}$ are \isolated when:
 \begin{align*}
  &v_{i}(t) \neq v_{j}(t), & 0 \leq t \leq \ell_{i}
  \\
  &v_{i}(t) \neq v_{j}(t-1) \land v_{i}(t-1) \neq v_{j}(t), & 0 < t \leq \ell_{i}
  \\
  &v_{i}(\ell_{i}) \neq v_{j}(t), & \ell_{i} + 1 \leq t \leq \ell_{j}
 \end{align*}
\end{definition}
\begin{definition}
  If all pairs of paths are \isolated, $\bm{\pi}$ is \safe.
\end{definition}

\begin{figure*}[t]
 \centering
 \begin{tikzpicture}
    \coordinate[](0-0) at (0, 0);
    \coordinate[above=0.9cm of 0-0](0-1);
    \coordinate[above=0.9cm of 0-1](0-2);
    \coordinate[above=0.9cm of 0-2](0-3);
    \coordinate[above=0.9cm of 0-3](0-4);
    \coordinate[above=0.9cm of 0-4](0-5);
    \coordinate[right=0.9cm of 0-0](1-0);
    \coordinate[above=0.9cm of 1-0](1-1);
    \coordinate[above=0.9cm of 1-1](1-2);
    \coordinate[above=0.9cm of 1-2](1-3);
    \coordinate[above=0.9cm of 1-3](1-4);
    \coordinate[above=0.9cm of 1-4](1-5);
    \coordinate[right=0.9cm of 1-0](2-0);
    \coordinate[above=0.9cm of 2-0](2-1);
    \coordinate[above=0.9cm of 2-1](2-2);
    \coordinate[above=0.9cm of 2-2](2-3);
    \coordinate[above=0.9cm of 2-3](2-4);
    \coordinate[above=0.9cm of 2-4](2-5);
    \coordinate[right=0.9cm of 2-0](3-0);
    \coordinate[above=0.9cm of 3-0](3-1);
    \coordinate[above=0.9cm of 3-1](3-2);
    \coordinate[above=0.9cm of 3-2](3-3);
    \coordinate[above=0.9cm of 3-3](3-4);
    \coordinate[above=0.9cm of 3-4](3-5);
    \coordinate[right=0.9cm of 3-0](4-0);
    \coordinate[above=0.9cm of 4-0](4-1);
    \coordinate[above=0.9cm of 4-1](4-2);
    \coordinate[above=0.9cm of 4-2](4-3);
    \coordinate[above=0.9cm of 4-3](4-4);
    \coordinate[above=0.9cm of 4-4](4-5);
    \coordinate[right=0.9cm of 4-0](5-0);
    \coordinate[above=0.9cm of 5-0](5-1);
    \coordinate[above=0.9cm of 5-1](5-2);
    \coordinate[above=0.9cm of 5-2](5-3);
    \coordinate[above=0.9cm of 5-3](5-4);
    \coordinate[above=0.9cm of 5-4](5-5);
    \coordinate[right=0.9cm of 5-0](6-0);
    \coordinate[above=0.9cm of 6-0](6-1);
    \coordinate[above=0.9cm of 6-1](6-2);
    \coordinate[above=0.9cm of 6-2](6-3);
    \coordinate[above=0.9cm of 6-3](6-4);
    \coordinate[above=0.9cm of 6-4](6-5);
    \coordinate[right=0.9cm of 6-0](7-0);
    \coordinate[above=0.9cm of 7-0](7-1);
    \coordinate[above=0.9cm of 7-1](7-2);
    \coordinate[above=0.9cm of 7-2](7-3);
    \coordinate[above=0.9cm of 7-3](7-4);
    \coordinate[above=0.9cm of 7-4](7-5);
    \coordinate[right=0.9cm of 7-0](8-0);
    \coordinate[above=0.9cm of 8-0](8-1);
    \coordinate[above=0.9cm of 8-1](8-2);
    \coordinate[above=0.9cm of 8-2](8-3);
    \coordinate[above=0.9cm of 8-3](8-4);
    \coordinate[above=0.9cm of 8-4](8-5);
    \coordinate[right=0.9cm of 8-0](9-0);
    \coordinate[above=0.9cm of 9-0](9-1);
    \coordinate[above=0.9cm of 9-1](9-2);
    \coordinate[above=0.9cm of 9-2](9-3);
    \coordinate[above=0.9cm of 9-3](9-4);
    \coordinate[above=0.9cm of 9-4](9-5);
    \coordinate[right=0.9cm of 9-0] (10-0);
    \coordinate[above=0.9cm of 10-0](10-1);
    \coordinate[above=0.9cm of 10-1](10-2);
    \coordinate[above=0.9cm of 10-2](10-3);
    \coordinate[above=0.9cm of 10-3](10-4);
    \coordinate[above=0.9cm of 10-4](10-5);
    \coordinate[right=0.9cm of 10-0](11-0);
    \coordinate[above=0.9cm of 11-0](11-1);
    \coordinate[above=0.9cm of 11-1](11-2);
    \coordinate[above=0.9cm of 11-2](11-3);
    \coordinate[above=0.9cm of 11-3](11-4);
    \coordinate[above=0.9cm of 11-4](11-5);
    \coordinate[right=0.9cm of 11-0](12-0);
    \coordinate[above=0.9cm of 12-0](12-1);
    \coordinate[above=0.9cm of 12-1](12-2);
    \coordinate[above=0.9cm of 12-2](12-3);
    \coordinate[above=0.9cm of 12-3](12-4);
    \coordinate[above=0.9cm of 12-4](12-5);
    \coordinate[right=0.9cm of 12-0](13-0);
    \coordinate[above=0.9cm of 13-0](13-1);
    \coordinate[above=0.9cm of 13-1](13-2);
    \coordinate[above=0.9cm of 13-2](13-3);
    \coordinate[above=0.9cm of 13-3](13-4);
    \coordinate[above=0.9cm of 13-4](13-5);
    \coordinate[right=0.9cm of 13-0](14-0);
    \coordinate[above=0.9cm of 14-0](14-1);
    \coordinate[above=0.9cm of 14-1](14-2);
    \coordinate[above=0.9cm of 14-2](14-3);
    \coordinate[above=0.9cm of 14-3](14-4);
    \coordinate[above=0.9cm of 14-4](14-5);
    \coordinate[right=0.9cm of 14-0](15-0);
    \coordinate[above=0.9cm of 15-0](15-1);
    \coordinate[above=0.9cm of 15-1](15-2);
    \coordinate[above=0.9cm of 15-2](15-3);
    \coordinate[above=0.9cm of 15-3](15-4);
    \coordinate[above=0.9cm of 15-4](15-5);
    \coordinate[right=0.9cm of 15-0](16-0);
    \coordinate[above=0.9cm of 16-0](16-1);
    \coordinate[above=0.9cm of 16-1](16-2);
    \coordinate[above=0.9cm of 16-2](16-3);
    \coordinate[above=0.9cm of 16-3](16-4);
    \coordinate[above=0.9cm of 16-4](16-5);
    \coordinate[right=0.9cm of 16-0](17-0);
    \coordinate[above=0.9cm of 17-0](17-1);
    \coordinate[above=0.9cm of 17-1](17-2);
    \coordinate[above=0.9cm of 17-2](17-3);
    \coordinate[above=0.9cm of 17-3](17-4);
    \coordinate[above=0.9cm of 17-4](17-5);
    \coordinate[right=0.45cm of 1-0](1_5-5);
    \coordinate[right=0.45cm of 3-0](3_5-5);
    \coordinate[right=0.45cm of 8-0](8_5-5);
    \coordinate[right=0.45cm of 13-0](13_5-5);
  \coordinate[right=0.45cm of 1-1](1_1-5);
  \coordinate[right=0.45cm of 3-1](3_1-5);
  \coordinate[right=0.45cm of 8-1](8_1-5);
  \coordinate[right=0.45cm of 13-1](13_1-5);
  \begin{scope}
   \coordinate[left=3.5cm of 0-0](tmp-0);
   \coordinate[below=0.5cm of tmp-0](zero);
   \node[vertex,label=below:{nodes}](v2) at (zero) {$v5$};
   \node[vertex, left=0.4cm of v2] (v0) {$v4$};
   \node[vertex,above=0.4cm of v0](v1) {$v1$};
   \node[vertex,above=0.4cm of v2](v3) {$v2$};
   \node[vertex,right=0.4cm of v2](v4) {$v6$};
   \node[vertex,above=0.4cm of v4](v5) {$v3$};
   \foreach \u / \v in {v0/v1,v2/v3,v4/v5,v0/v2,v2/v4,v1/v3,v3/v5}
   \draw[line] (\u) -- (\v);
  \end{scope}
  \begin{scope}
   \coordinate[left=3.5cm of 0-0](tmp-0);
   \coordinate[below=3.1cm of tmp-0](zero);
   \node[vertex,label=below:{initial state ($t=0$)}](v2) at (zero) {};
   \node[agent0, left=0.4cm of v2,label=left:{\scriptsize \tiny{priority:} \textbf{high}}](v0) {$a_{1}$};
   \node[vertex,above=0.4cm of v0](v1) {};
   \node[agent1,above=0.4cm of v2,label=above:{\scriptsize medium-high}](v3) {$a_{2}$};
   \node[agent3,right=0.4cm of v2,label=right:{\scriptsize low}](v4) {$a_{4}$};
   \node[agent2,above=0.4cm of v4,label=right:{\scriptsize medium-low}](v5) {$a_{3}$};
   \foreach \u / \v in {v0/v1,v2/v3,v4/v5,v0/v2,v2/v4,v1/v3,v3/v5}
   \draw[line] (\u) -- (\v);
  \end{scope}
  \begin{scope}
   \coordinate[left=3.5cm of 0-0](tmp-0);
   \coordinate[below=5.5cm of tmp-0](zero);
   \node[agent1,label=below:{goal state}](v2) at (zero) {$a_{2}$};
   \node[vertex, left=0.4cm of v2](v0) {};
   \node[agent3,above=0.4cm of v0](v1) {$a_{4}$};
   \node[vertex,above=0.4cm of v2](v3) {};
   \node[agent2,right=0.4cm of v2](v4) {$a_{3}$};
   \node[agent0,above=0.4cm of v4](v5) {$a_{1}$};
   \foreach \u / \v in {v0/v1,v2/v3,v4/v5,v0/v2,v2/v4,v1/v3,v3/v5}
   \draw[line] (\u) -- (\v);
  \end{scope}
    \begin{scope} 
     \coordinate[above=0.5cm of 1-0](time-s);
     \coordinate[below=6.75cm of time-s](time-g);
     \draw[line,->](time-s)--(time-g);
     \node[label=left:{$t=0$},below=0.55cm of time-s](t0) {\textbf{-}};
     \node[label=left:{$t=1$},below=2.3cm of time-s](t1) {\textbf{-}};
     \node[label=left:{$t=2$},below=4.05cm of time-s](t2) {\textbf{-}};
     \node[label=left:{$t=3$},below=5.8cm of time-s](t3) {\textbf{-}};
     \node[label=left:{(for all)},above=0.2cm of t0]() {};
     \coordinate[above=0.5cm of 5-0](phase-1-s);
     \coordinate[below=6.75cm of phase-1-s](phase-1-g);
     \draw[line,dashed](phase-1-s)--(phase-1-g);
     \coordinate[above=0.5cm of 9-0](phase-2-s);
     \coordinate[below=6.75cm of phase-2-s](phase-2-g);
     \draw[line,dashed](phase-2-s)--(phase-2-g);
    \end{scope}
    \begin{scope}  
     \coordinate[right=0.05cm of 3-0](3_5-0);
     \node[vertex,below=0.4cm of 3_5-0](v2-0) {};
     \node[agent0, left=0.4cm of v2-0](v0-0){$a_{1}$};
     \node[vertex,above=0.4cm of v0-0](v1-0) {};
     \node[agent1,above=0.4cm of v2-0](v3-0) {$a_{2}$};
     \node[agent3,right=0.4cm of v2-0](v4-0) {$a_{4}$};
     \node[agent2,above=0.4cm of v4-0](v5-0) {$a_{3}$};
     \node[dest](d0) at (v2-0) {};
     \foreach \u / \v in {v0-0/v1-0,v2-0/v3-0,v4-0/v5-0,v0-0/v2-0,v2-0/v4-0,v1-0/v3-0,v3-0/v5-0}
     \draw[line] (\u) -- (\v);
     \draw[pi] (v0-0)--(d0);
     \node[above=-0.05cm of v3-0](){$\ell_1=0$};
     \node[vertex,below=1.15cm of v0-0](v0-1) {};
     \node[vertex,  above=0.4cm of v0-1](v1-1) {};
     \node[agent0-d,right=0.4cm of v0-1](v2-1) {$a_{1}$};
     \node[vertex,  above=0.4cm of v2-1](v3-1) {};
     \node[vertex,  right=0.4cm of v2-1](v4-1) {};
     \node[vertex,  above=0.4cm of v4-1](v5-1) {};
     \foreach \u / \v in {v0-1/v1-1,v2-1/v3-1,v4-1/v5-1,v0-1/v2-1,v2-1/v4-1,v1-1/v3-1,v3-1/v5-1}
     \draw[line] (\u) -- (\v);
     \node[vertex,below=1.15cm of v0-1](v0-2) {};
     \node[vertex,above=0.4cm of v0-2](v1-2) {};
     \node[vertex,right=0.4cm of v0-2](v2-2) {};
     \node[vertex,above=0.4cm of v2-2](v3-2) {};
     \node[vertex,right=0.4cm of v2-2](v4-2) {};
     \node[vertex,above=0.4cm of v4-2](v5-2) {};
     \foreach \u / \v in {v0-2/v1-2,v2-2/v3-2,v4-2/v5-2,v0-2/v2-2,v2-2/v4-2,v1-2/v3-2,v3-2/v5-2}
     \draw[line] (\u) -- (\v);
     \node[vertex,below=1.15cm of v0-2](v0-3) {};
     \node[vertex,above=0.4cm of v0-3](v1-3) {};
     \node[vertex,right=0.4cm of v0-3](v2-3) {};
     \node[vertex,above=0.4cm of v2-3](v3-3) {};
     \node[vertex,right=0.4cm of v2-3](v4-3) {};
     \node[vertex,above=0.4cm of v4-3](v5-3) {};
     \foreach \u / \v in {v0-3/v1-3,v2-3/v3-3,v4-3/v5-3,v0-3/v2-3,v2-3/v4-3,v1-3/v3-3,v3-3/v5-3}
     \draw[line] (\u) -- (\v);
    \end{scope}
    \begin{scope}  
     \coordinate[right=0.05cm of 7-0](8_5-0);
     \node[vertex,below=0.4cm of 8_5-0](v2-0) {};
     \node[agent0, left=0.4cm of v2-0](v0-0){$a_{1}$};
     \node[vertex,above=0.4cm of v0-0](v1-0) {};
     \node[agent1,above=0.4cm of v2-0](v3-0) {$a_{2}$};
     \node[agent3,right=0.4cm of v2-0](v4-0) {$a_{4}$};
     \node[agent2,above=0.4cm of v4-0](v5-0) {$a_{3}$};
     \node[dest](d3) at (v5-0) {};
     \node[dest](d2) at (v3-0) {};
     \node[dest](d1) at (v1-0) {};
     \foreach \u / \v in {v0-0/v1-0,v2-0/v3-0,v4-0/v5-0,v0-0/v2-0,v2-0/v4-0,v1-0/v3-0,v3-0/v5-0}
     \draw[line] (\u) -- (\v);
     \draw[pi](v4-0)--(d3);
     \draw[pi](v5-0)--(d2);
     \draw[pi](v3-0)--(d1);
     \node[above=-0.05cm of v3-0](){$\ell_1=1$};
     \node[vertex,below=1.15cm of v0-0](v0-1) {};
     \node[agent1-d,above=0.4cm of v0-1](v1-1) {$a_{2}$};
     \node[agent0,  right=0.4cm of v0-1](v2-1) {$a_{1}$};
     \node[agent2-d,above=0.4cm of v2-1](v3-1) {$a_{3}$};
     \node[vertex,  right=0.4cm of v2-1](v4-1) {};
     \node[agent3-d,above=0.4cm of v4-1](v5-1) {$a_{4}$};
     \node[dest](d0) at (v4-1) {};
     \foreach \u / \v in {v0-1/v1-1,v2-1/v3-1,v4-1/v5-1,v0-1/v2-1,v2-1/v4-1,v1-1/v3-1,v3-1/v5-1}
     \draw[line] (\u) -- (\v);
     \draw[pi](v2-1)--(d0);
     \node[vertex,below=1.15cm of v0-1](v0-2) {};
     \node[vertex,  above=0.4cm of v0-2](v1-2) {};
     \node[vertex,  right=0.4cm of v0-2](v2-2) {};
     \node[vertex,  above=0.4cm of v2-2](v3-2) {};
     \node[agent0-d,right=0.4cm of v2-2](v4-2) {$a_{1}$};
     \node[vertex,  above=0.4cm of v4-2](v5-2) {};
     \foreach \u / \v in {v0-2/v1-2,v2-2/v3-2,v4-2/v5-2,v0-2/v2-2,v2-2/v4-2,v1-2/v3-2,v3-2/v5-2}
     \draw[line] (\u) -- (\v);
     \node[vertex,below=1.15cm of v0-2](v0-3) {};
     \node[vertex,above=0.4cm of v0-3](v1-3) {};
     \node[vertex,right=0.4cm of v0-3](v2-3) {};
     \node[vertex,above=0.4cm of v2-3](v3-3) {};
     \node[vertex,right=0.4cm of v2-3](v4-3) {};
     \node[vertex,above=0.4cm of v4-3](v5-3) {};
     \foreach \u / \v in {v0-3/v1-3,v2-3/v3-3,v4-3/v5-3,v0-3/v2-3,v2-3/v4-3,v1-3/v3-3,v3-3/v5-3}
     \draw[line] (\u) -- (\v);
    \end{scope}
    \begin{scope}
     \coordinate[right=0.05cm of 11-0](13_5-0);
     \node[vertex,below=0.4cm of 13_5-0](v2-0) {};
     \node[agent0, left=0.4cm of v2-0](v0-0){$a_{1}$};
     \node[vertex,above=0.4cm of v0-0](v1-0) {};
     \node[agent1,above=0.4cm of v2-0](v3-0) {$a_{2}$};
     \node[agent3,right=0.4cm of v2-0](v4-0) {$a_{4}$};
     \node[agent2,above=0.4cm of v4-0](v5-0) {$a_{3}$};
     \foreach \u / \v in {v0-0/v1-0,v2-0/v3-0,v4-0/v5-0,v0-0/v2-0,v2-0/v4-0,v1-0/v3-0,v3-0/v5-0}
     \draw[line] (\u) -- (\v);
     \node[above=-0.05cm of v3-0](){$\ell_1=2$};
     \node[vertex,below=1.15cm of v0-0](v0-1) {};
     \node[agent1,above=0.4cm of v0-1](v1-1) {$a_{2}$};
     \node[agent0,right=0.4cm of v0-1](v2-1) {$a_{1}$};
     \node[agent2,above=0.4cm of v2-1](v3-1) {$a_{3}$};
     \node[vertex,right=0.4cm of v2-1](v4-1) {};
     \node[agent3,above=0.4cm of v4-1](v5-1) {$a_{4}$};
     \node[dest](d3) at (v3-1) {};
     \node[dest](d2) at (v2-1) {};
     \foreach \u / \v in {v0-1/v1-1,v2-1/v3-1,v4-1/v5-1,v0-1/v2-1,v2-1/v4-1,v1-1/v3-1,v3-1/v5-1}
     \draw[line] (\u) -- (\v);
     \draw[pi](v5-1)--(d3);
     \draw[pi](v3-1)--(d2);
     \node[vertex,below=1.15cm of v0-1](v0-2) {};
     \node[vertex,  above=0.4cm of v0-2](v1-2) {};
     \node[agent2-d,right=0.4cm of v0-2](v2-2) {$a_{3}$};
     \node[agent3-d,above=0.35cm of v2-2](v3-2) {$a_{4}$};
     \node[agent0,  right=0.4cm of v2-2](v4-2) {$a_{1}$};
     \node[vertex,  above=0.4cm of v4-2](v5-2) {};
     \node[dest](d0) at (v5-2) {};
     \foreach \u / \v in {v0-2/v1-2,v2-2/v3-2,v4-2/v5-2,v0-2/v2-2,v2-2/v4-2,v1-2/v3-2,v3-2/v5-2}
     \draw[line] (\u) -- (\v);
     \draw[pi](v4-2)--(d0);
     \node[vertex,below=1.15cm of v0-2](v0-3) {};
     \node[vertex,  above=0.4cm of v0-3](v1-3) {};
     \node[vertex,  right=0.4cm of v0-3](v2-3) {};
     \node[vertex,  above=0.4cm of v2-3](v3-3) {};
     \node[vertex,  right=0.4cm of v2-3](v4-3) {};
     \node[agent0-d,above=0.4cm of v4-3](v5-3) {$a_{1}$};
     \foreach \u / \v in {v0-3/v1-3,v2-3/v3-3,v4-3/v5-3,v0-3/v2-3,v2-3/v4-3,v1-3/v3-3,v3-3/v5-3}
     \draw[line] (\u) -- (\v);
    \end{scope}
 \end{tikzpicture}
 \caption{
 Example of \winpibt.
 The left part shows the configuration.
 The right part illustrates how agents secure nodes sequentially.
 Here, the window size is three.
 Due to the space limit, we show the example that $a_{1}$ secures its path until $t=3$.
 In the right part, agents who have already secured nodes are depicted by color-filled circles.
 Requested nodes are depicted by edge-colored circles.
 We do not illustrate backtracking to avoid complexity.
 In this example, $a_{1}$ has initiative.
 First, $a_{1}$ tries to secure $v5$ at $t=1$ (first column).
 This attempt succeeds since it does not violate others locations.
 Next, $a_1$ tries to secure $v6$ at $t=2$ (second column).
 This attempt will break $\bm{\pi}$ from \safe.
 Thus, priority inheritance occurs from $a_1$ to $a_4$.
 Now $a_4$ has initiative and plans a single-step path.
 For $a_i$ such that $\ell_i = \ell_4$, this action works as same as PIBT.
 After this, identically, priority inheritance is executed between related agents.
 }
 \label{fig:winpibt-example}
\end{figure*}
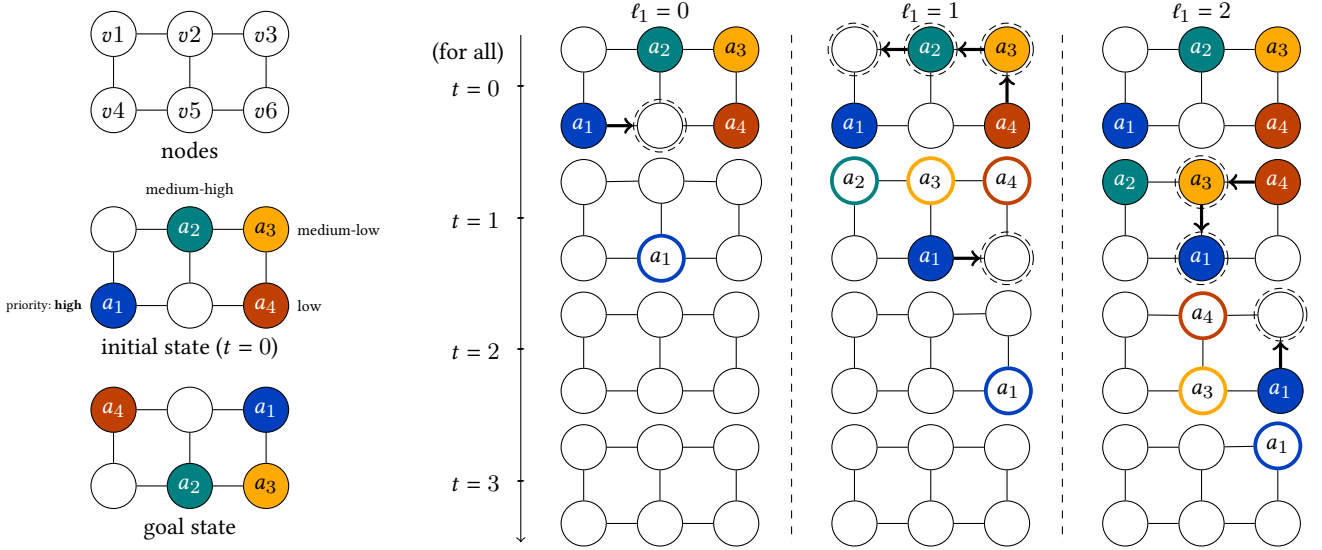

From the definition of \safe condition, it is trivial that when $\bm{\pi}$ is \safe, agents do not collide until timestep $\min \{ \ell_{i} \}$.
Moreover, a combination of extending paths exists such that agents do not ever collide.
\begin{proposition}
 If $\bm{\pi}$ is \safe, for $a_{i}$, there exists at least one additional path until any timestep $t$ ($t \geq \ell_{i}$) while keeping $\bm{\pi}$ \safe.
 \label{lemma:keep-safety}
\end{proposition}
\begin{proof}
  Make $a_{i}$ stay its last assigned location until timestep $t$.
  This operation obviously keeps $\bm{\pi}$ \safe.
\end{proof}

The \safe condition might be helpful to developing online solvers by regarding as temporal terminate condition.
In online situations, where goals are dynamically assigned to agents, the challenge is replanning paths on demand.
One intuitive but excessive approach is to update paths for all agents in the system until a certain timestep, e.g., Replan All~\cite{vsvancara2019online}, and this certainly ensures conflict-free.
The disentangled condition can relax this type of replanning, i.e., it enables to update paths for part of agents, and still ensure the safety.

PIBT can be understood as making an effort to keep $\bm{\pi}$ \safe.
Priority inheritance occurs when $a_{i}$ attempts to break the \isolated condition regarding $\pi_{i}$ and $\pi_{j}$, where $a_{j}$ is an agent with lower priority.
Then $a_j$ secures the next node so as to keep $\pi_i$ and $\pi_j$ are \isolated, before $a_i$ does.
In strictly, there is one exception: movements corresponding to rotations.
Assume that $a_i, a_j, a_k$ tries to move $v_j(t), v_k(t), v_i(t)$, respectively.
If $a_i$ has the highest priority, $a_j$ secures the node prior to $a_i$, and $a_k$ does prior to $a_j$.
$\pi_k$ and $\pi_i$ are not \isolated temporary, but $\bm{\pi}$ revives in \safe immediately since rotations always succeed in PIBT.
\winpibt works as same as PIBT, i.e., update paths while keeping $\bm{\pi} \safe$.
The difference is that \winpibt can perform priority inheritance retroactively.

\section{Windowed PIBT (\winpibt)}
\label{sec:algo}
In this section, we first provide a basic concept of how to extend the time window of PIBT and an example.
Then, the pseudo code is given with theoretical analysis.
We explain \winpibt in centralized fashion.
PIBT itself is a relatively realistic approach for decentralized implementation, however, \winpibt with decentralized fashion faces some difficulties as discussed later.

\subsection{Concept}
Similarly to PIBT, \winpibt makes the agent with highest priority move along an arbitrary path within a time window.
The original PIBT algorithm plans paths for all agents one by one timestep, i.e., PIBT relies on a unit-length time window.
\winpibt extends the time window of PIBT while satisfying Lemma~\ref{lemma:pibt-local-movement}.
Describing simply, the algorithm for one agent $a_{i}$ consists of three phases:
\begin{itemize}[leftmargin=0.4cm]
\item[1)] Compute a path ideal for $a_{i}$ that excludes already reserved time-node pairs while avoiding interference with the progression of higher-priority agents.
\item[2)] Secure time-node pairs sequentially along to the computed path.
\item[3)] If the node requested at $t_{i}$ is the last assigned node for some agent~$a_{j}$ at $t_{j}$ such that $t_{j} < t_{i}$, i.e., $\pi_i$ and $\pi_j$ will not be \isolated, then move $a_j$ from the node by priority inheritance.
  In precisely, let $a_j$ plan its path one step ahead by inheriting the priority of $a_i$ until there are no such agents.
  If such an agent remains until $t_i$, then $a_i$ executes the PIBT algorithm with the property of Lemma~\ref{lemma:pibt-local-movement}.
\end{itemize}

\subsubsection{Example}
Fig.~\ref{fig:winpibt-example} illustrates how \winpibt works.
To simplify, we remove the invalid case of priority inheritance.
Here, $a_1$ has the highest priority and it takes initiative.
Assume that the window size is three.
At the beginning, $a_1$ computes the ideal path $(v4, v5, v6, v3)$ and starts securing nodes.
$v5$ at $t=1$ can be regarded as ``unoccupied'' since the last allocated nodes for the other agents are $v2$, $v3$ and $v6$.
Thus, $a_{1}$ secures $v5$ at $t=1$.
Next, $a_{1}$ tries to secure $v6$ at $t=2$ that is the last assigned node of $a_{4}$, i.e., $v_{4}(\ell_{4})$.
$a_{1}$ has to compel $a_{4}$ to move from $v6$ before $t=2$ and priority inheritance occurs between different timesteps (from $a_{1}$ at $t=1$ to $a_{4}$ at $t=0$).
This inheritance process continues until $a_{2}$ secures the node via $a_{3}$ and $a_{4}$, just like in PIBT.
Now $a_{2}$, $a_{3}$ and $a_{4}$ secure the nodes until $t=1$.
This causes $v6$ at $t=2$ to become ``unoccupied'' and hence $a_{1}$ successfully secures the desired node.
The above process continues until the initiative agent reserves the nodes at the current timestep ($t=0$) plus the window size (3).
After $a_{1}$ finishes reservation, $a_{2}$ now starts reservation from $t=2$ avoiding the already secured node, e.g., $v2$ at $t=2$ cannot be used since it is already assigned to $a_4$ (to make space for $a_1$).
Finally, \winpibt gives the paths as follows:
\begin{itemize}
\item $\pi_{1}$: $(v4, v5, v6, v3)$
\item $\pi_{2}$: $(v2, v1, v4, v5)$
\item $\pi_{3}$: $(v3, v2, v5, v6)$
\item $\pi_{4}$: $(v6, v3, v2, v1)$
\end{itemize}

\subsection{Algorithm}
We show pseudo code of \winpibt in Algorithm~\ref{algo:func-winpibt} and \ref{algo:caller}.
The former describes function $\mathsf{winPIBT}$ that gives $a_{i}$ a path until $t=\alpha$.
The latter shows how to call function $\mathsf{winPIBT}$ globally.
\winpibt has a recursive structure with respect to priority inheritance and backtracking similarly to PIBT.

Function $\mathsf{winPIBT}$ takes four arguments:
1)~$a_{i}$ is an agent determining its own path;
2)~$\alpha$ is timestep until which $a_{i}$ secures nodes, i.e., $\ell_i = \alpha$ after calling function $\mathsf{winPIBT}$;
3)~$\bm{\Pi}$ represents provisional paths of all agents.
Each agent plans its own path while referring to $\bm{\Pi}$.
We denote by $\Pi_{i}$ a provisional path of $a_{i}$ and $\Pi_{i}(t)$ a node at timestep $t$ in $\Pi_{i}$.
Intuitively, $\Pi_{i}$ consists of connecting an already determined path $\pi_{i}$ and a path trying to secure.
Note that $0 \leq t \leq \ell_{i}$, $\Pi_{i}(t) = v_{i}(t)$;
4)~$R$ is a set of agents which are currently requesting some nodes, aiming at detecting rotations.
In the pseudo code, we also implicitly use $\ell_{j}$, which is not contained in arguments.

We use three functions:
1-2) $\mathsf{validPath}(a_{i}, \beta, \bm{\Pi})$ and $\mathsf{registerPath}$
$(a_{i}, \alpha, \beta, \bm{\Pi})$ compute a path for $a_{i}$.
The former confirms whether there exists a path for $a_{i}$ such that keeps $\bm{\pi}$ \safe from timestep $\ell_{i} + 1$ to $\beta$.
The latter computes the ideal path until timestep $\beta$ and registers it to $\bm{\Pi}$ until $t=\alpha$.
We assume that always $\alpha \leq \beta$.
In addition to prohibiting collisions, $\Pi_{i}$ is constrained by the following term;
$\Pi_{i}(t_{i}) \neq \Pi_{j}(t_{j}), \ell_{i} < t_{i} < \ell_{j}, t_{i} < t_{j} \leq \ell_{j}$.
Intuitively, this constraint says that the shorter path cannot invade the longer path.
The rationale is to keep $\bm{\Pi}$ \safe.
In \winpibt, a path is elongated by adding nodes one by one to its end.
Assume two paths with different lengths.
The \safe condition is broken in two cases: The longer path adds the last node of the shorter path to its end, or, the shorter path adds the node that the longer path uses in the gap term between two paths.
The constraint prohibits the latter case.
A critical example, shown in Fig.~\ref{fig:winpibt-reservation:cannotenter},
assumes the following situation:
After $a_1$ has fixed its path, $a_2$ starts securing nodes.
To do so, $a_2$ tries to secure the current location of $a_3$, thus, $a_3$ has to plan its path.
What happens when $a_3$ plans to use the crossing node to avoid a temporal collision with $a_2$?
The problem is that $a_3$ is not ensured to return to its first location since $a_2$ has a higher priority.
Thus, an agent with a lower priority has the potential to be stuck on the way of a path of an agent with higher priority without this constraint.
According to this constraint, $a_3$ cannot use a crossing node until $a_1$ passes.
This can cause some problematic cases as shown in Fig.~\ref{fig:winpibt-reservation:inconvenient}, which implies that \emph{extra reservation leads to awkward path planning}.
3) $\mathsf{copeStuck}(a_{i}, \alpha, \bm{\Pi})$ is called when $a_{i}$ has no path satisfying the constraints.
This forcibly gives a path to $a_{i}$ such that staying at the last assigned node until timestep $\alpha$, i.e., $v_{i}(\ell_{i}+1),\dots,v_{i}(\alpha) \leftarrow v_{i}(\ell_{i})$.
Note that this function also updates $\Pi_i$ such that $\Pi_i = \pi_i$.

Algorithm~\ref{algo:func-winpibt} is as follows.
An agent $a_{i}$ enters a path decision phase when function $\mathsf{winPIBT}$ is called with the first argument $a_{i}$.
Firstly, it checks that the timestep when the last node was assigned to $a_{i}$ is smaller than $\alpha$, or else the path of $a_{i}$ has already been determined over $t=\alpha$, thus $\mathsf{winPIBT}$ returns as valid [Line~\ref{algo:func-winpibt:init-check}].
Next, it computes the prophetic timestep $\beta$ [Line~\ref{algo:func-winpibt:calc-beta}].
The rationale of $\beta$ is that,
unless $a_{i}$ sends backtracking, the provisional paths in $\bm{\Pi}$
that may affect the planning of $a_i$ never change.
Thus, computing a path based on an upper timestep $\beta$ works akin to forecasting.
If no valid path exist, $a_{i}$ is forced to stay at $v_{i}(\ell_{i})$ until $t=\alpha$ via the function $\mathsf{copingStuck}$ and backtracks as invalid [Line~\ref{algo:func-winpibt:validpath-1}--\ref{algo:func-winpibt:end-valid-path-1}].
A similar operation is executed when $a_{i}$ recomputes its path [Line~\ref{algo:func-winpibt:validpath-2}--\ref{algo:func-winpibt:invalid-2}].
After that, $\mathsf{winPIBT}$ proceeds: 
1)~Compute an ideal path for $a_{i}$ satisfying the constraints [Line~\ref{algo:func-winpibt:register1}, \ref{algo:func-winpibt:register2}];
2)~Secure time-node pairs sequentially along path $\Pi_{i}$ [Line~\ref{algo:func-winpibt:target-node}, \ref{algo:func-winpibt:secure-node}];
3)~If the requesting node $v$ is violating a path of $a_{j}$, let $a_{j}$ leave $v$ by $t=\ell_{i}-1$ via priority inheritance [Line~\ref{algo:func-winpibt:force-agent}--\ref{algo:func-winpibt:end-force-agent}].
If any agent $a_{j}$ remains at $t=\ell_{i}$, then the original PIBT works [Line~\ref{algo:func-winpibt:require-pibt}--\ref{algo:func-winpibt:end-force}].
Note that introducing $R$ prevents eternal priority inheritance and enables rotations.
\begin{algorithm}[t]
 {\small
 \caption{function $\mathsf{winPIBT}$}
 \label{algo:func-winpibt}
 \begin{algorithmic}
  \Require
  \State $a_{i}$ : an agent trying to secure nodes.
  \State $\alpha$ : timestep by which $a_{i}$ secures. After calling, $\ell_{i} = \alpha$.
  \State $\bm{\Pi}$ : provisional paths.
  \State $R$: agents that are currently requesting some nodes.
  \Ensure $\{\text{valid}, \text{invalid}\}$
 \end{algorithmic}
 \vspace{0.01cm}
 \begin{algorithmic}[1]
  \Function{$\mathsf{winPIBT}$}
  {$a_{i}, \alpha, \bm{\Pi}$, $R$}
  \IfSingle{$\ell_{i} \geq \alpha$}{\Return valid}
  \label{algo:func-winpibt:init-check}
  \State $\beta \leftarrow \max(\alpha, \text{max timestep registered in } \bm{\Pi})$
  \label{algo:func-winpibt:calc-beta}
  \If{$\nexists \mathsf{validPath}(a_{i}, \beta, \bm{\Pi})$}
  \label{algo:func-winpibt:validpath-1}
  \State $\mathsf{copeStuck}(a_{i}, \alpha, \bm{\Pi})$
  \label{algo:func-winpibt:copestuck-1}
  \State \Return invalid
  \label{algo:func-winpibt:invalid-1}
  \EndIf
  \label{algo:func-winpibt:end-valid-path-1}
  \State $\mathsf{registerPath}(a_{i}, \alpha, \beta, \bm{\Pi})$
  \label{algo:func-winpibt:register1}
  \State $R \leftarrow R \cup \{ a_{i} \}$
  \label{algo:func-winpibt:init-path}
  \While{$\ell_{i} < \alpha$}
  \label{algo:func-winpibt:while-start}
  \State $v \leftarrow \Pi_{i}(\ell_{i} + 1)$
  \Comment target node
  \label{algo:func-winpibt:target-node}
  \While{$\exists a_{j}$ s.t. $\ell_{j} < \ell_{i}, v_{j}(\ell_{j}) = v$}
  \label{algo:func-winpibt:force-agent}
  \State $\mathsf{winPIBT} (a_{j}, \ell_{j} + 1, \bm{\Pi}, R)$
  \label{algo:func-winpibt:let-aj-force}
  \EndWhile
  \label{algo:func-winpibt:end-force-agent}
  \If{$\exists a_{j}$ s.t. $\ell_{j} = \ell_{i}, v_{j}(\ell_{j}) = v, a_{j} \not\in R$}
  \label{algo:func-winpibt:require-pibt}
  \If{$\mathsf{winPIBT} (a_{j}, \ell_{j} + 1, \bm{\Pi}, R)$ \small{is invalid}}
  \label{algo:func-winpibt:pibt}
  \State revoke nodes $\Pi_{i}(\ell_{i} + 1), \dots, \Pi_{i}(\alpha)$
  \If{$\nexists \mathsf{validPath}(a_{i}, \beta, \bm{\Pi})$}
  \label{algo:func-winpibt:validpath-2}
  \State $\mathsf{copeStuck}(a_{i}, \alpha, \bm{\Pi})$
  \label{algo:func-winpibt:copestuck-2}
  \State \Return invalid
  \label{algo:func-winpibt:invalid-2}
  \Else
  \State $\mathsf{registerPath}(a_{i}, \alpha, \beta, \bm{\Pi})$
  \label{algo:func-winpibt:register2}
  \State \textbf{continue}
  \EndIf
  \label{algo:func-winpibt:end-recompute}
  \EndIf
  \EndIf
  \label{algo:func-winpibt:end-force}
  \State $v_{i}(\ell_{i} + 1) \leftarrow v$
  \Comment $\ell_{i}$ is implicitly incremented
  \label{algo:func-winpibt:secure-node}
  \EndWhile
  \label{algo:func-winpibt:endwhile}
  \State \Return valid
  \label{algo:func-winpibt:bt-valid}
  \EndFunction
 \end{algorithmic}
}
\end{algorithm}
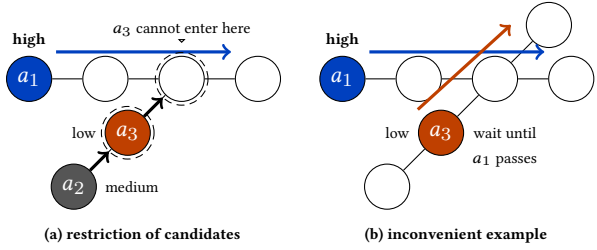
\begin{figure}[t]
 \centering
 \begin{tabular}{cc}
  \begin{minipage}[t]{0.45\hsize}
   \centering
   \begin{tikzpicture}
    \node[agenthigh,label=above:{\scriptsize \textbf{high}}](v0) at (0,0) {$a_1$};
    \node[vertex,right=0.4cm of v0](v1) {};
    \node[vertex,right=0.4cm of v1](v2) {};
    \node[vertex,right=0.4cm of v2](v3) {};
    \node[agentlow,below left=0.4cm of v2,label=left:{\scriptsize low}](v5) {$a_3$};
    \node[agent,below left=0.4cm of v5,label=right:{\scriptsize medium}](v6) {$a_2$};
    \foreach \u / \v in {v0/v1,v1/v2,v2/v3,v2/v5,v5/v6}
    \draw[line] (\u) -- (\v);
    \coordinate[above right=0.2cm of v0](a0-s);
    \coordinate[above left =0.2cm of v3](a0-g);
    \draw[move-high](a0-s)--(a0-g);
    \node[dest](d2) at (v2) {};
    \node[dest](d1) at (v5) {};
    \draw[pi](v6)--(d1);
    \draw[pi](v5)--(d2);
    \node[above=0.15cm of v2](col) {\scriptsize $a_3$ cannot enter here};
    \coordinate[above=0.18cm of v2](tri1);
    \coordinate[above left=0.05cm of tri1](tri2);
    \coordinate[above right=0.05cm of tri1](tri3);
    \draw[line](tri1)--(tri2)--(tri3)--cycle;
   \end{tikzpicture}
   \subcaption{restriction of candidates}
   \label{fig:winpibt-reservation:cannotenter}
  \end{minipage}
  &
  \begin{minipage}[t]{0.45\hsize}
   \centering
   \begin{tikzpicture}
    \node[agenthigh,label=above:{\scriptsize \textbf{high}}](v0) at (0,0) {$a_1$};
    \node[vertex,right=0.4cm of v0](v1) {};
    \node[vertex,right=0.4cm of v1](v2) {};
    \node[vertex,right=0.4cm of v2](v3) {};
    \node[vertex,above right=0.4cm of v2](v4) {};
    \node[agentlow,below left=0.4cm of v2,label=left:{\scriptsize low}](v5) {$a_3$};
    \node[vertex,below left=0.4cm of v5](v6) {};
    \foreach \u / \v in {v0/v1,v1/v2,v2/v3,v2/v4,v2/v5,v5/v6}
    \draw[line] (\u) -- (\v);
    \coordinate[above right=0.2cm of v0](a0-s);
    \coordinate[above left =0.2cm of v3](a0-g);
    \draw[move-high](a0-s)--(a0-g);
    \coordinate[above left=0.13cm of v5](a2-s);
    \coordinate[left=0.2cm of v4](a2-g);
    \draw[move-low](a2-s)--(a2-g);
    \node[right=0.01cm of v5](a2-wait-1) {\scriptsize wait until};
    \node[below=0.cm of a2-wait-1](a2-wait-2) {\scriptsize $a_1$ passes};
   \end{tikzpicture}
   \subcaption{inconvenient example}
   \label{fig:winpibt-reservation:inconvenient}
  \end{minipage}
 \end{tabular}
 \caption{Reservation by \winpibt.
 Assume that $a_3$ enter the crossing node in Fig.~\ref{fig:winpibt-reservation:cannotenter} caused by priority inheritance from $a_2$.
 Following the \winpibt algorithm, $a_1$ has already fixed its path.
 Thereby, $a_3$ violates $a_1$'s progression.
 To avoid such situations, agents with lower priorities cannot enter the locations secured by agents with higher priorities until they passes.
 This results in inconvenient path sometimes (\ref{fig:winpibt-reservation:inconvenient}).
 }
 \label{fig:winpibt-reservation}
\end{figure}
\begin{algorithm}[t]
 {\small
 \caption{caller of function $\mathsf{winPIBT}$}
 \label{algo:caller}
 \begin{algorithmic}
  \State $t$: timestep
  \State $p_{i}(t)$ : priority of $a_{i}$ at timestep $t$. $p_{i}(t) \in \mathbb{R}$
  \State $w_{i}(t)$ : window size of $a_{i}$ at timestep $t$. $w_{i}(t) \in \mathbb{N}$
  \State $\kappa$: maximum timestep that agents can secure nodes, initially $0$
  \State $\bm{\Pi}$: provisional paths, initially $\Pi_{i} = ( v_{i}(0) )$
 \end{algorithmic}
 \begin{algorithmic}[1]
  \State {\tiny \textit{(for every timestep $t$)}}
  \State task allocation if required
  \Comment i.e., update goals
  \State update all priorities $p_{i}(t)$, windows $w_{i}(t)$
  \label{algo:caller:update}
  \State Let $\mathcal{U}$ denote a sorted list of $A$ by $p_{i}$
  \For{$ j \in 1, \dots, n$}
  \State $a_{i} \leftarrow \mathcal{U}[j]$
  \If{$\ell_{i} \leq t$}
  \label{algo:caller:skip}
  \If{$j = 1$}
  \Comment the agent with highest priority
  \State $\mathsf{winPIBT}\left(a_{i}, t + w_{i}(t), \bm{\Pi}, \emptyset\right)$
  \label{algo:caller:call1}
  \Else
  \State $\mathsf{winPIBT}\left(a_{i}, \min(t + w_{i}(t), \kappa), \bm{\Pi}, \emptyset\right)$
  \label{algo:caller:call2}
  \EndIf
  \EndIf
  \State $\kappa \leftarrow \min(\kappa, \ell_{i})$
  \IfSingle{$j = 1$}{$\kappa \leftarrow \ell_{i}$}
  \EndFor
 \end{algorithmic}
}
\end{algorithm}

There is a little flexibility on how to call function $\mathsf{winPIBT}$.
Algorithm~\ref{algo:caller} shows one example.
In each timestep~$t$ before the path adjustment phase, the priority $p_{i}(t)$ of an agent $a_{i}$ is updated as mentioned later [Line~\ref{algo:caller:update}].
The window $w_{i}(t)$ is also updated [Line~\ref{algo:caller:update}].
In this paper, we fix $w_{i}(t)$ to be a constant value.
Then, agents elongate their own paths in order of priorities.
Agents that have already determined their path until the current timestep~$t$, i.e., $\ell_{i} > t$, skip making a path [Line~\ref{algo:caller:skip}].
In order not to disturb paths of agents with higher priorities, an upper bound on timesteps $\kappa$ is introduced.
By $\kappa$, agents with lower priorities are prohibited to update the paths beyond lengths of paths of agents with higher priorities.

By the next lemma, \winpibt always gives valid paths.
\begin{lemma}
 \winpibt keeps $\bm{\pi}$ \safe.
 \label{lemma:winpibt-safety}
\end{lemma}
\begin{proof}
 Initially, $\bm{\pi}$ is \safe.
 $\bm{\pi}$ is updated via the function $\mathsf{winPIBT}$.
 Before an agent $a_{i}$ calculates a path, $a_{i}$ confirms the existence of a path from timestep $\ell_{i} + 1$ until $\beta$, as defined in Line~\ref{algo:func-winpibt:init-check} while avoiding collisions and using $v$ such that $\ell_i < t^{\prime} \leq \ell_j, v_j(t^{\prime}) = v$, with respect to paths registered in $\bm{\Pi}$.
 We distinguish two cases: 1)~a path exists, or, 2) no path exists.

 \begin{enumerate}[leftmargin=0.5cm]
  \item[1)] a path exists:
 $a_{i}$ now successfully computes a path $\Pi_{i}$ satisfying the condition and starts securing nodes accordingly.
 Assume that $a_{i}$ tries to secure a node $v$ at timestep $\gamma = \ell_{i} + 1$.
 We distinguish three cases regarding other agents $a_{j}$ and their $\ell_{j}$.
 \begin{enumerate}
  \item[a.] $\ell_{j} > \ell_{i}$:
            $\Pi_{i}$ is computed without collisions with paths on $\bm{\Pi}$.
            Moreover, $\Pi_{i}$ avoids $v$ at $t = \gamma$ such that $\forall t^\prime, \gamma < t^\prime \leq \ell_{j}, v_{j}(t^\prime) = v$.
            Thus, $\pi_{i}$ and $\pi_{j}$ are \isolated if $a_{i}$ adds $v$ in its path $\pi_{i}$.
  \item[b.] $\ell_{j} < \ell_{i}$:
            If $v_{j}(\ell_{j}) \neq v$, the operation of adding $v$ to $\pi_{i}$ keeps $\pi_{i}$ and $\pi_{j}$ \isolated, or else $a_{i}$ tries to let $a_{j}$ leave $v$ via priority inheritance [Line~\ref{algo:func-winpibt:let-aj-force}].
            $a_{j}$ now gets the privilege to determine $v_{j}(\ell_{j} + 1)$.
            This action of $a_{j}$ remains $\pi_{i}$ and $\pi_{j}$ \isolated following two reasons.
            First, $a_{i}$ never secures $v$ until $a_{j}$ goes away.
            Second, if $a_{j}$ successfully computes a path $\Pi_{j}$, the previous part applies.
            If failed, $v_{j}(\ell_{j} + 1)$ is set to $v_{j}(\ell)$, i.e., $v$.
            This action also keeps $\pi_{i}$ and $\pi_{j}$ \isolated.
            If some agent $a_{j}$ stays on $v$ until timestep $\ell_{i}$, the following part applies.
  \item[c.] $\ell_{j} = \ell_{i}$:
            This case is equivalent to the PIBT algorithm.
            If $v \neq v_{j}(\ell_{j})$, the operation of adding $v$ to $\pi_{i}$ keeps $\pi_{i}$ and $\pi_{j}$ \isolated, or else there are two possibilities: either $a_{j} \not\in R$ or $a_{j} \in R$.
            If $a_{j} \not\in R$, $a_{j}$ inherits the priority of $a_{i}$.
            When the outcome of backtracking is valid, this means that, at timestep $\gamma$, $a_{j}$ secures a node other than $v$ and $v_{i}(\ell_{i})$ (to avoid swap conflict), since both have already registered in $\Pi_{i}$.
            Thus, $a_{i}$ successfully secures $v$ while keeping $\pi_{i}$ and $\pi_{j}$ \isolated.
            When the outcome is invalid, $a_{j}$ stays at its current node, i.e., $v_{j}(\gamma) = v_{j}(\gamma - 1)$, and $a_{i}$ recomputes $\Pi_{i}$.
            Still, $\pi_{i}$ and $\pi_{j}$ are \isolated since $a_{i}$ has not secured a node at timestep $\gamma$.
            Next, consider the case where $a_{j} \in R$.
            This happens when $a_{j}$ is currently requesting another node.
            Thus, after $a_{i}$ secures $v$ at timestep $\gamma$ and backtracking returns as valid, $a_{j}$ successfully secures the node.
            $\pi_{i}$ and $\pi_{j}$ is temporally not \safe, but $\pi_{i}$ recovers the $\safe$ condition immediately.
            Intuitively, this case corresponds to rotations.
 \end{enumerate}
 As a result, $\bm{\pi}$ is kept \safe through the action of $a_{i}$ to secure a node.

 \item[2)] no path exists: In this case, $a_{i}$ chooses to stay at its current node.
 Obviously, this action keeps $\bm{\pi}$ \safe.
\end{enumerate}

Therefore, regardless of whether a path exists or not, $\bm{\pi}$ is \safe.
\end{proof}

The following lemma shows that the agent with highest priority can move arbitrarily, akin to Lemma~\ref{lemma:pibt-local-movement} for PIBT.
\begin{lemma}
 $\mathsf{winPIBT}(a_{i}, \alpha, \bm{\Pi}, \emptyset)$ gives $a_{i}$ an arbitrary path from $t=\ell_i+1$ until $t = \alpha$ if $G$ is \graphcond, and $\forall j \neq i, \ell_{i} \geq \ell_{j}$ and $\bm{\Pi} = \bm{\pi}$.
 \label{lemma:winpibt-highest}
\end{lemma}
\begin{proof}
 $\forall j \neq i, |\Pi_{j}| \leq |\Pi_{i}|$ holds since $\ell_{j} \leq \ell_{i}$ and $\bm{\Pi} = \bm{\pi}$.
 Thus, $a_{i}$ can compute an arbitrary path $\Pi_{i}$ from timestep $\ell_{i} + 1$ until $\alpha$.
 We now show that $a_{i}$ never receives invalid as outcome of backtracking.
 According to $\Pi_{i}$, $a_{i}$ tries to secure a node sequentially.
 Let $v$ be this node at timestep $\gamma$.
 If $\not\exists a_{j}$ s.t. $v_{j}(\ell_{j}) = v$, $a_{i}$ obviously secures $v$ at timestep $\gamma$.
 The issue is only when $\exists a_{j}$ s.t. $v_{j}(\ell_{j})$, $\ell_{j} = \gamma - 1$, however, the equivalent mechanism of Lemma~\ref{lemma:pibt-local-movement} works and $a_{i}$ successfully moves to $v$ due to the assumption that $G$ is \graphcond.
 Thus, $a_{i}$ never receives an invalid outcome and moves on an arbitrary path until $t=\alpha$.
\end{proof}

\subsubsection{Prioritization}
The prioritization scheme of \winpibt is exactly the same as in the PIBT algorithm.
Let $\eta_{i}(t) \in \mathbb{N}$ be the timesteps elapsed since $a_{i}$ last updated the destination $g_{i}$ prior to timestep $t$.
Note that $\eta_{i}(0) = 0$.
Let $\epsilon_{i} \in [0,1)$ be a unique value to each agent $a_{i}$.
At every timestep, $p_{i}(t)$ is computed as the sum of $\eta_{i}(t)$ and $\epsilon_{i}$.
Thus, $p_{i}(t)$ is unique between agents in any timestep.

By this prioritization, we derive the following theorem.
\begin{theorem}
 By \winpibt, all agents reach their own destinations in finite timesteps after the destinations are given if $G$ is \graphcond and $\forall i, w_{i}(t)$ is kept finite in any timestep.
 \label{theorem:global-movement}
\end{theorem}
\begin{proof}
 Once $a_{i}$ gets the highest priority, the condition satisfying Lemma~\ref{lemma:winpibt-highest} comes true in finite timestep since no agents can newly reserve the path over the timestep limit set by the previous highest agent.
 Once such condition is realized, $a_{i}$ can move along the shortest path thanks to Lemma~\ref{lemma:winpibt-highest}.
 Until $a_{i}$ reaches its destination, this situation continues since
 Algorithm~\ref{algo:caller} ensures that function $\mathsf{winPIBT}$ is always called by other agents such that the second argument is never over $\ell_{i}$.
 Thus, $a_{i}$ reaches its destination in finite steps, and then drops its priority.
 During this, other agents increase their priority based on the definition of $\eta_{j}(t)$ and one of them obtains the highest priority after $a_{i}$ drops its priority.
 As long as such agents remain, the above-mentioned process is repeated.
 Therefore, all agents must reach their own destination in finite timestep after the destinations are given.
\end{proof}

\begin{figure*}[t]
 \centering
 \begin{tabular}{l}
  \begin{minipage}{0.24\hsize}
   \centering
   \begin{tikzpicture}
    \node[anchor=south west,inner sep=0] at (0,0)
    {\includegraphics[width=1\hsize]{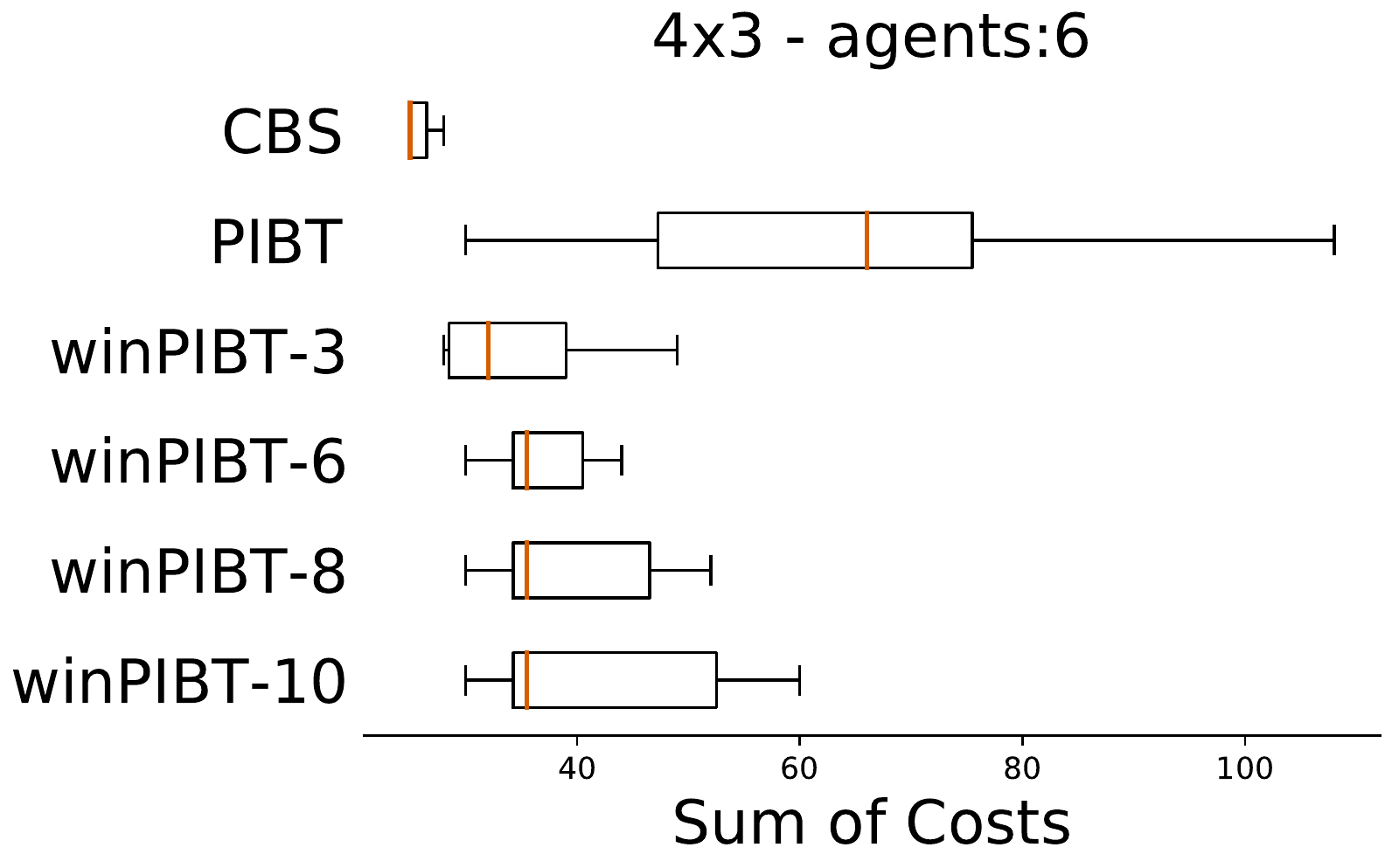}};
    \node[anchor=south west,inner sep=0] at (3.3, 0.5)
    {\includegraphics[width=0.18\hsize]{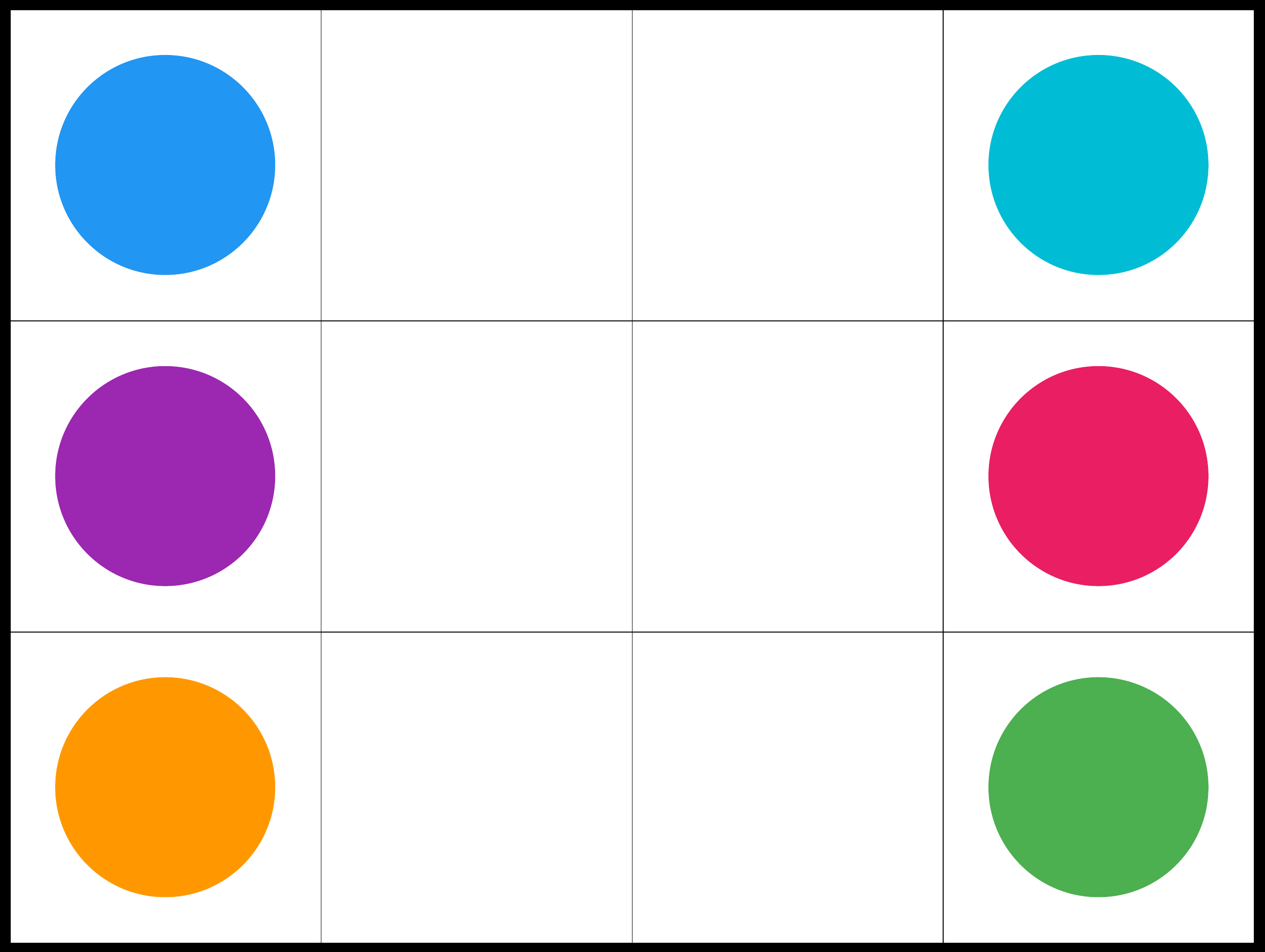}};
   \end{tikzpicture}
  \end{minipage}
  \begin{minipage}{0.24\hsize}
   \centering
   \begin{tikzpicture}
    \node[anchor=south west,inner sep=0] at (0,0)
    {\includegraphics[width=1\hsize]{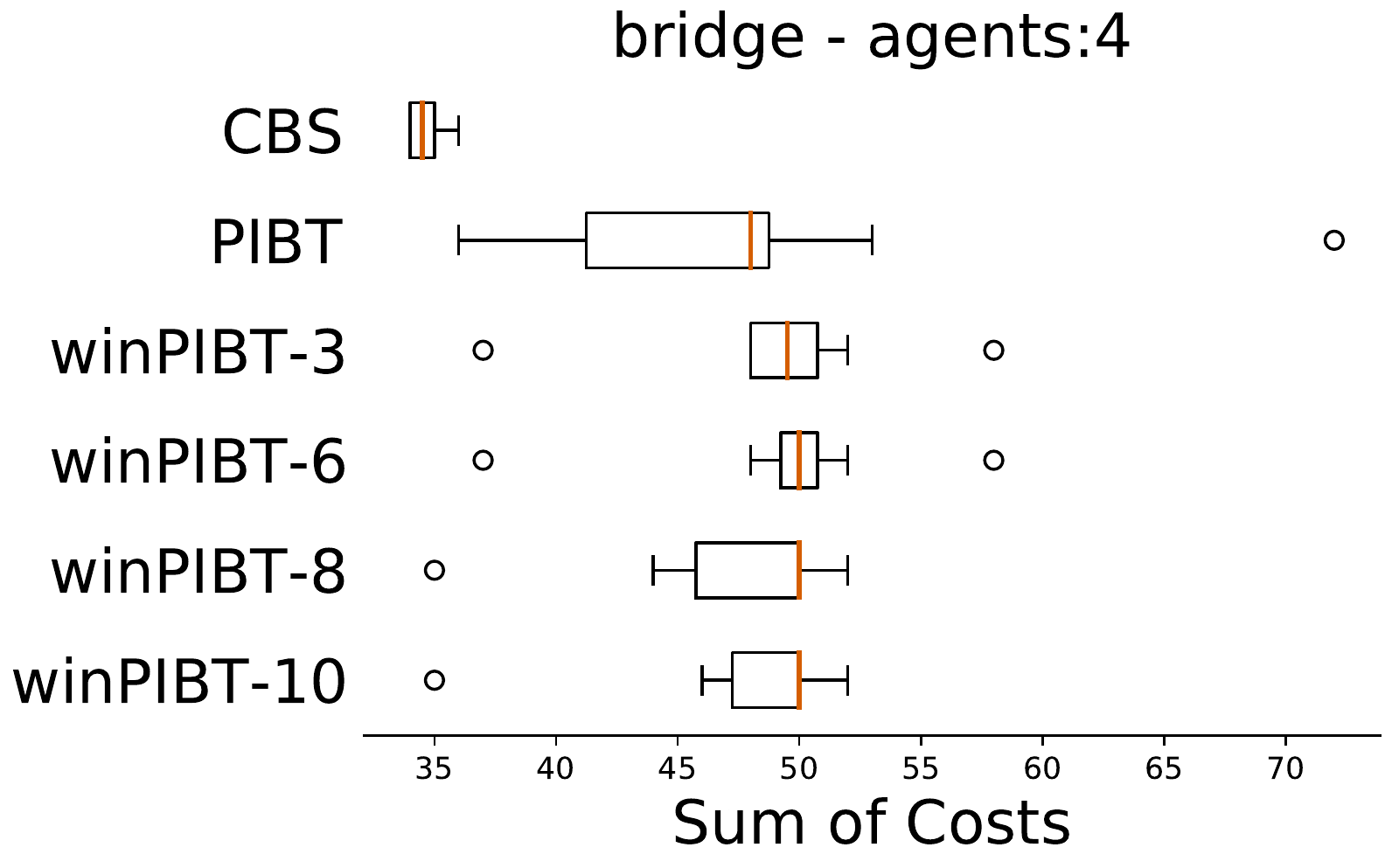}};
    \node[anchor=south west,inner sep=0] at (3.1, 0.5)
    {\includegraphics[width=0.25\hsize]{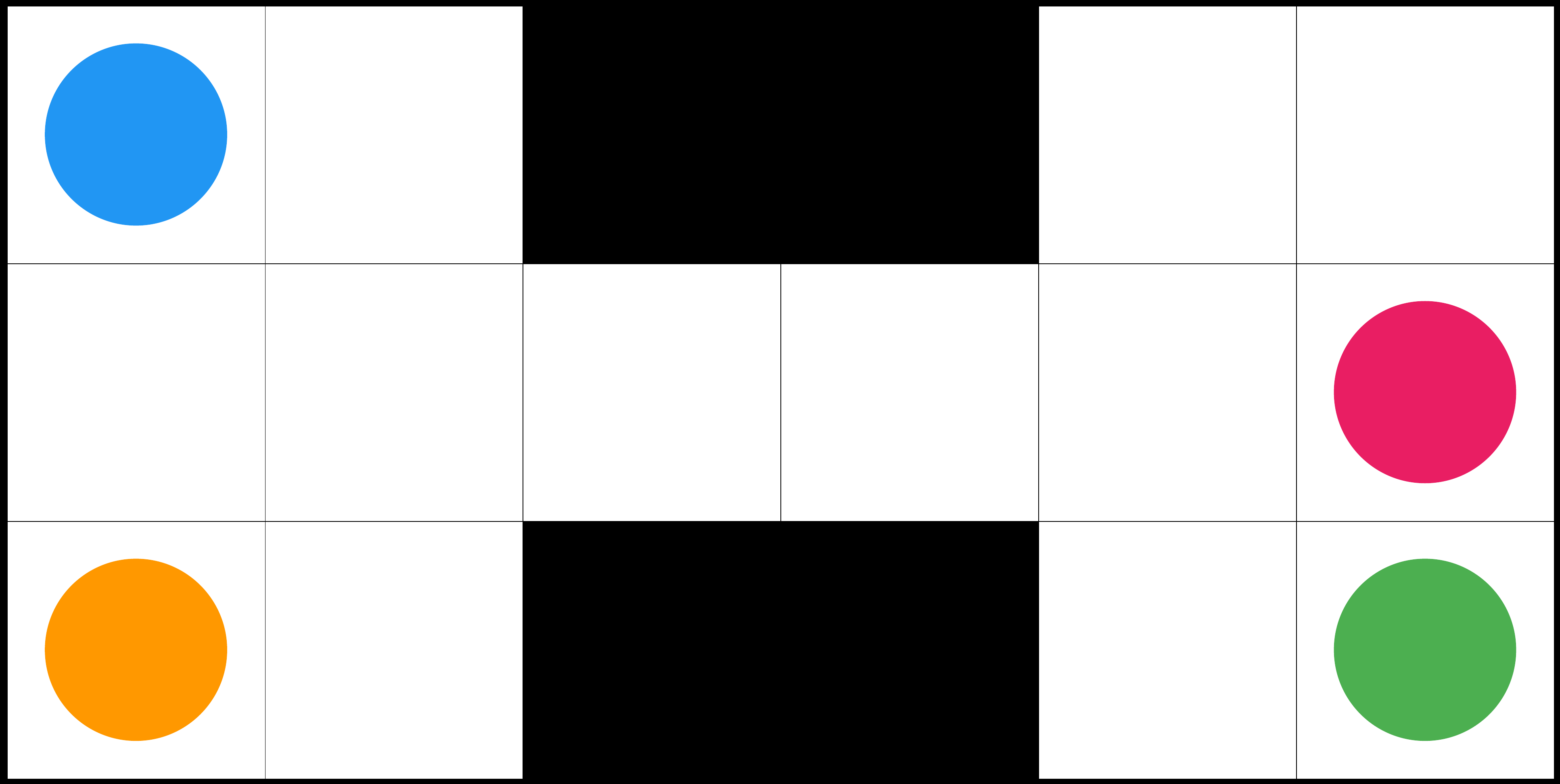}};
   \end{tikzpicture}
  \end{minipage}
  \begin{minipage}{0.24\hsize}
   \centering
   \begin{tikzpicture}
    \node[anchor=south west,inner sep=0] at (0,0)
    {\includegraphics[width=1\hsize]{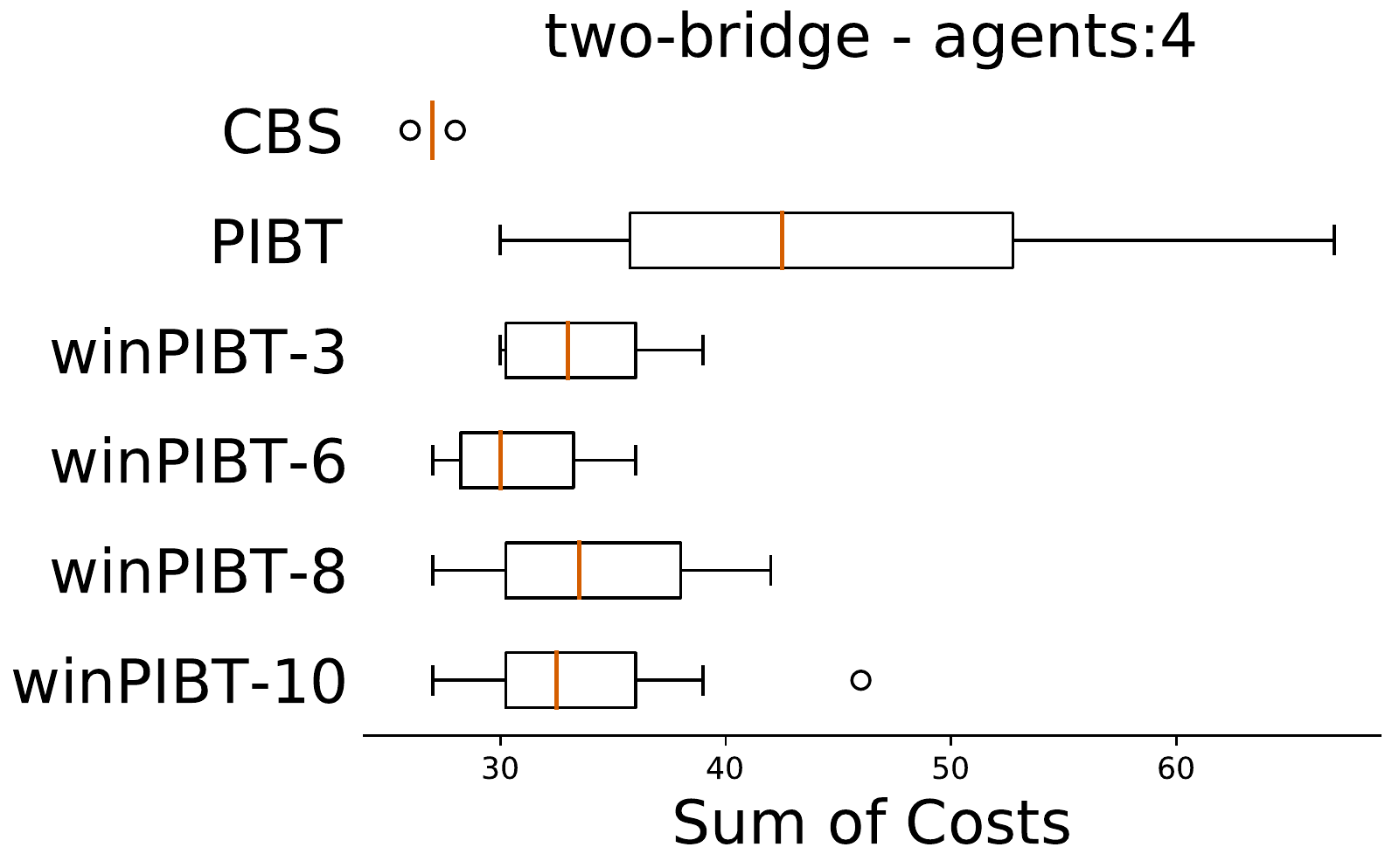}};
    \node[anchor=south west,inner sep=0] at (3.1, 0.5)
    {\includegraphics[width=0.25\hsize]{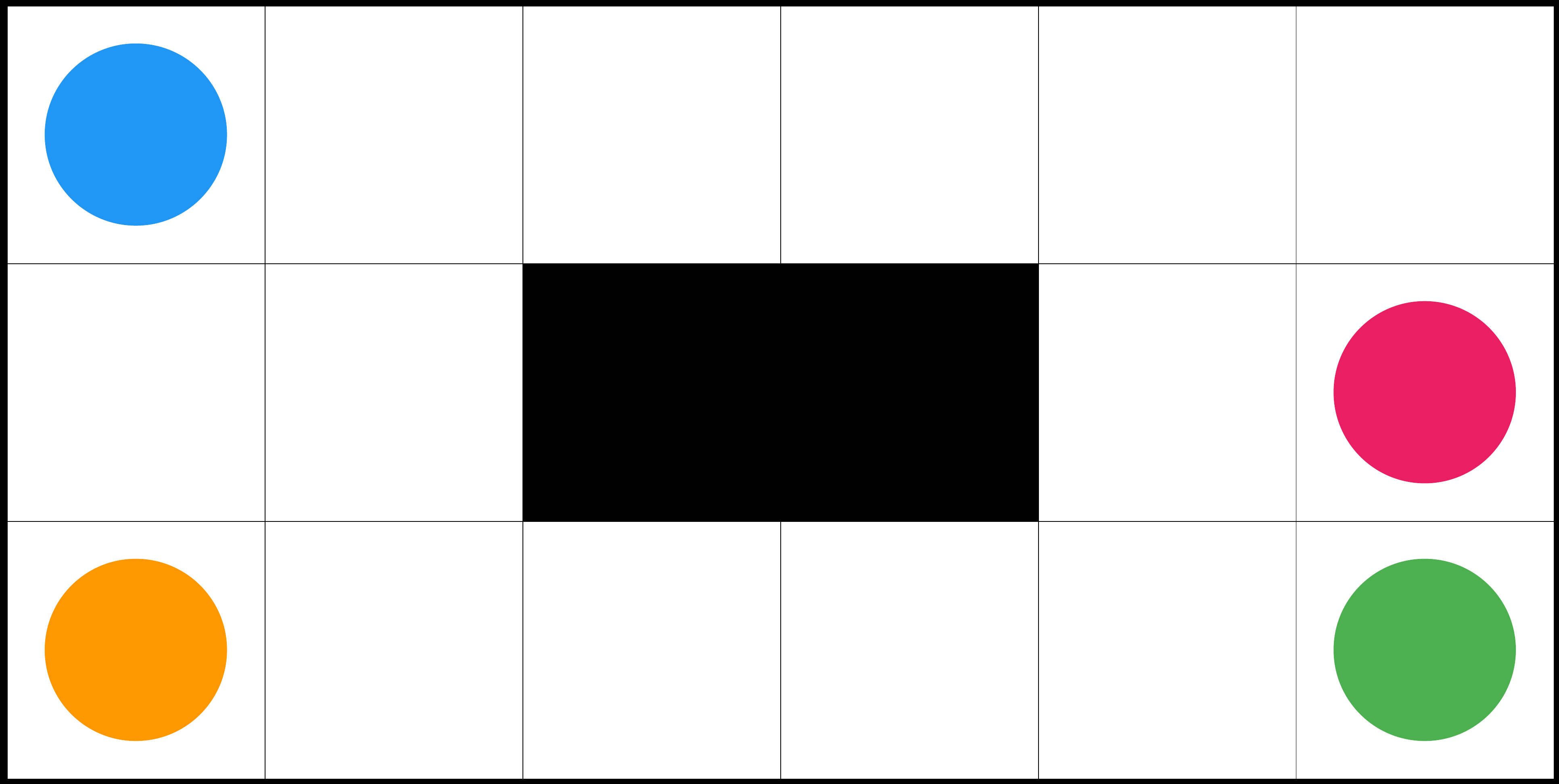}};
   \end{tikzpicture}
  \end{minipage}
  \begin{minipage}[]{0.24\hsize}
   \centering
   \begin{tikzpicture}
    \node[anchor=south west,inner sep=0] at (0,0)
    {\includegraphics[width=1.05\hsize]{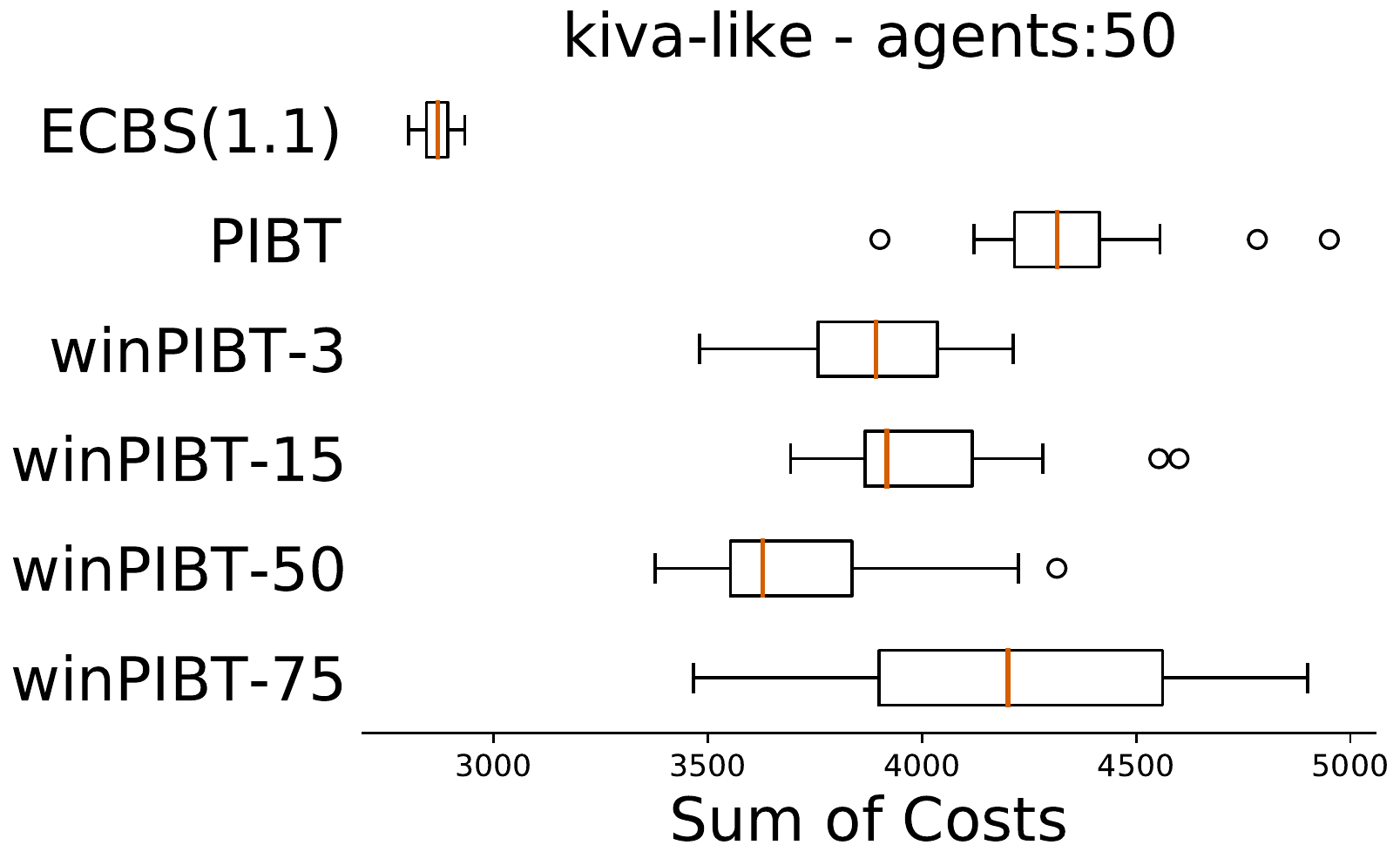}};
   \end{tikzpicture}
  \end{minipage}
 \end{tabular}
 \caption{Results of classical MAPF with basic benchmark.
   Benchmark are shown in each figure.
   Note that, the horizontal axes do not start from 0 to highlight the difference between solvers.}
 \label{fig:mapf-result-1}
\end{figure*}
\begin{figure*}[t]
 \centering
 \begin{tabular}{rllll}
  \begin{minipage}{0.08\hsize}
   \centering
   {\scriptsize \emptymid}
   \vspace{0.1cm}
   \includegraphics[width=1\hsize]{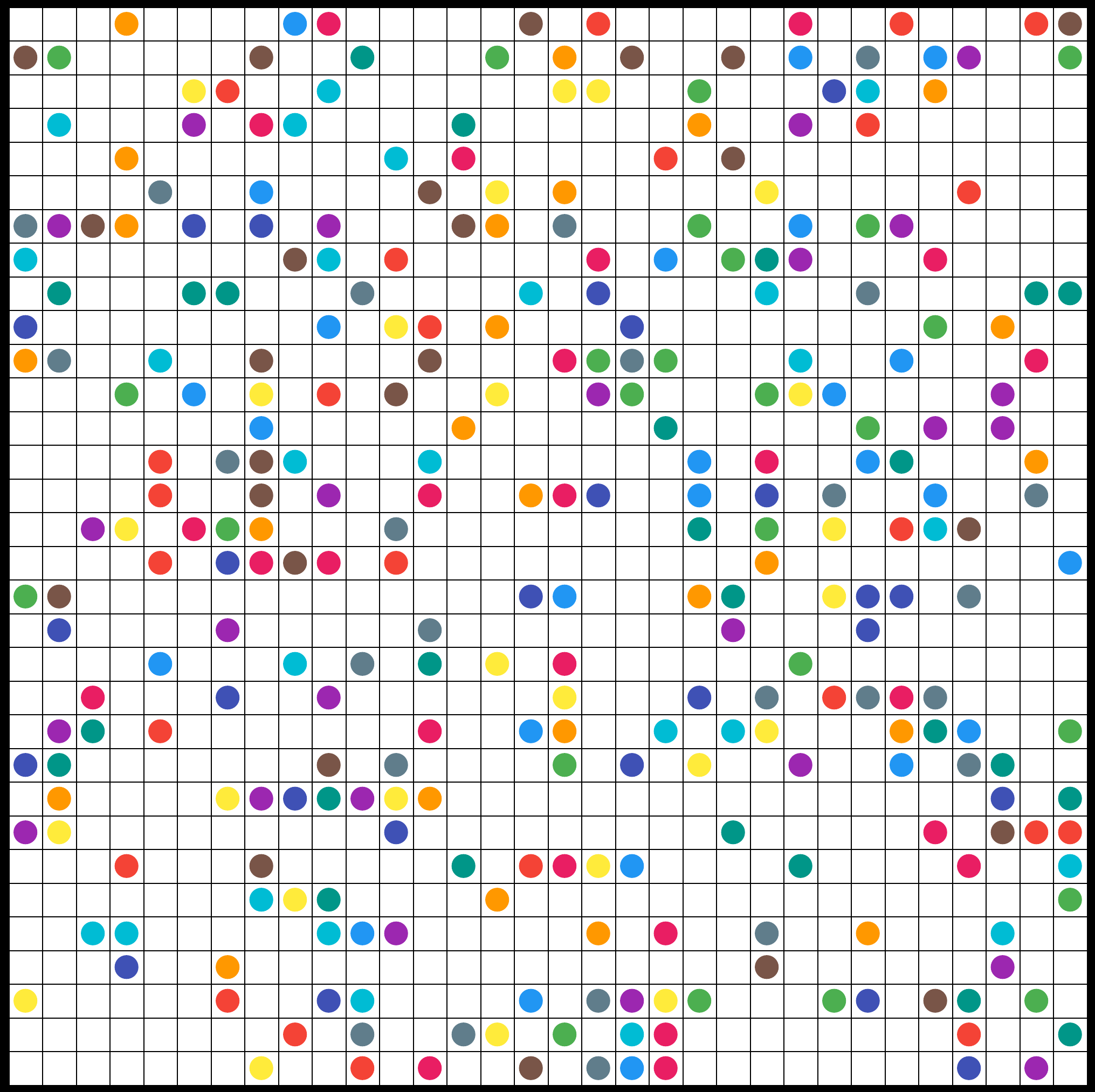}
  \end{minipage}
  \begin{minipage}{0.22\hsize}
   \centering
   \includegraphics[width=1\hsize]{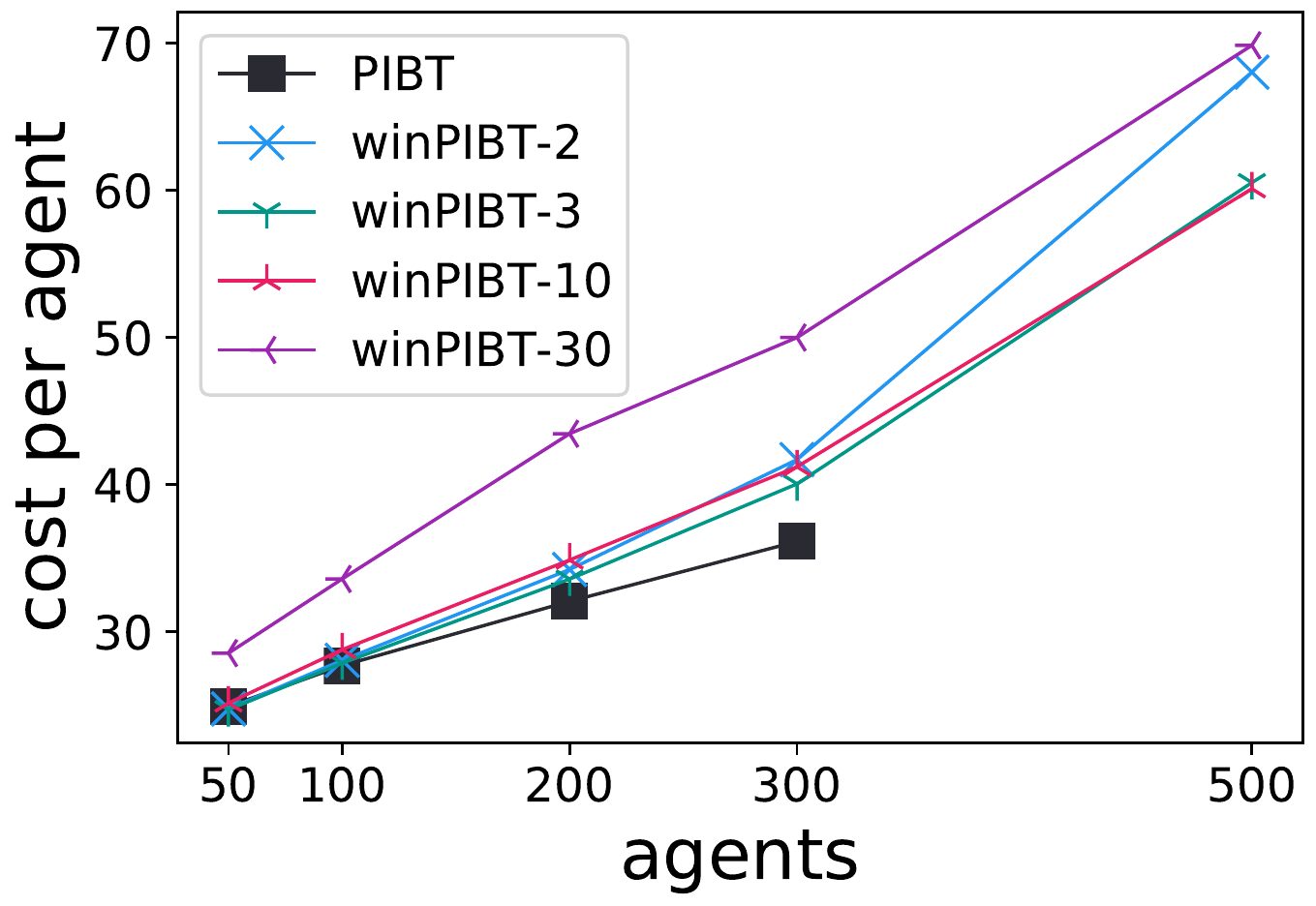}
  \end{minipage}
  \begin{minipage}{0.22\hsize}
   \centering
   \includegraphics[width=1\hsize]{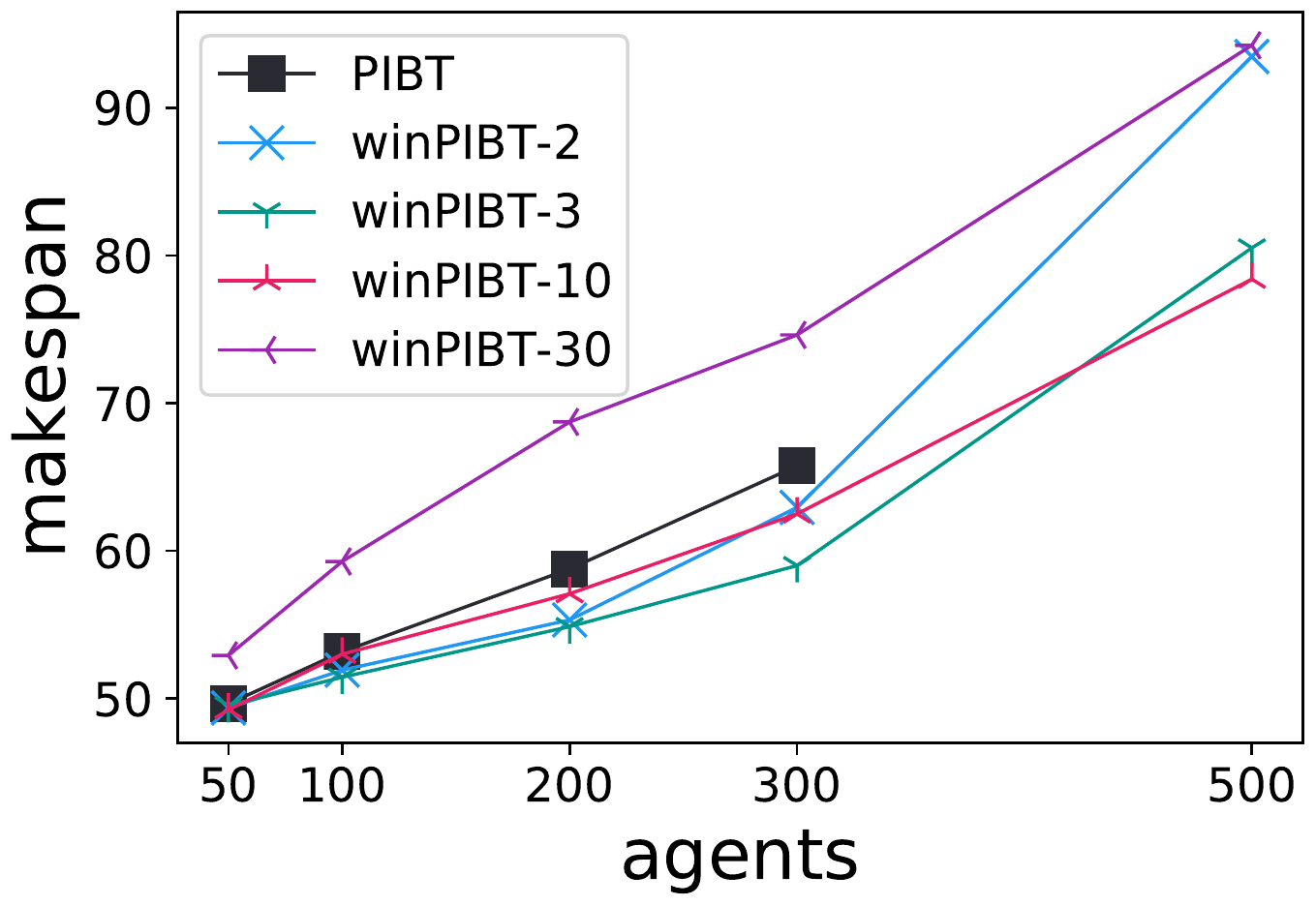}
  \end{minipage}
  \begin{minipage}{0.22\hsize}
   \centering
   \includegraphics[width=1\hsize]{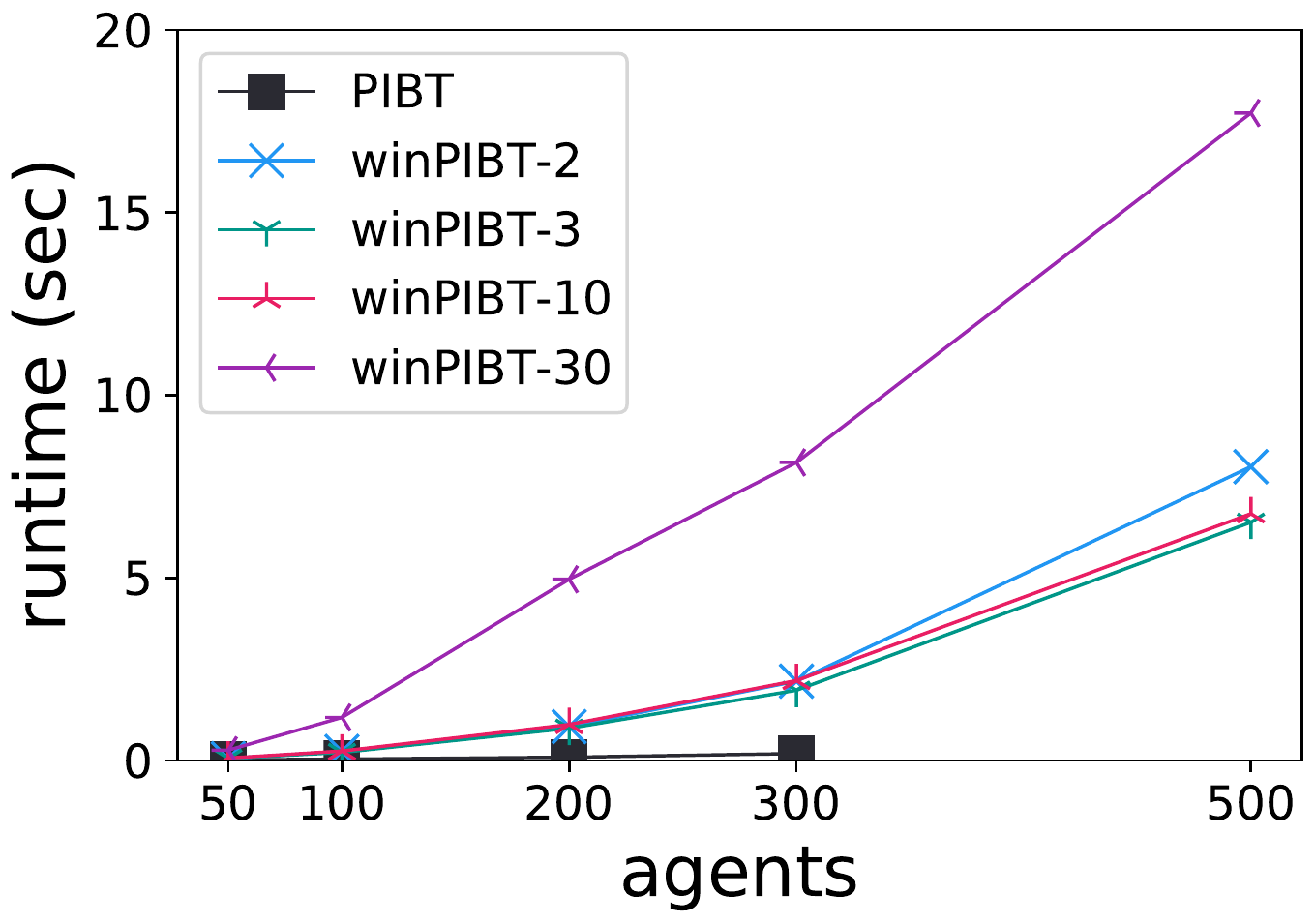}
  \end{minipage}
  \begin{minipage}{0.22\hsize}
   \centering
   \includegraphics[width=1\hsize]{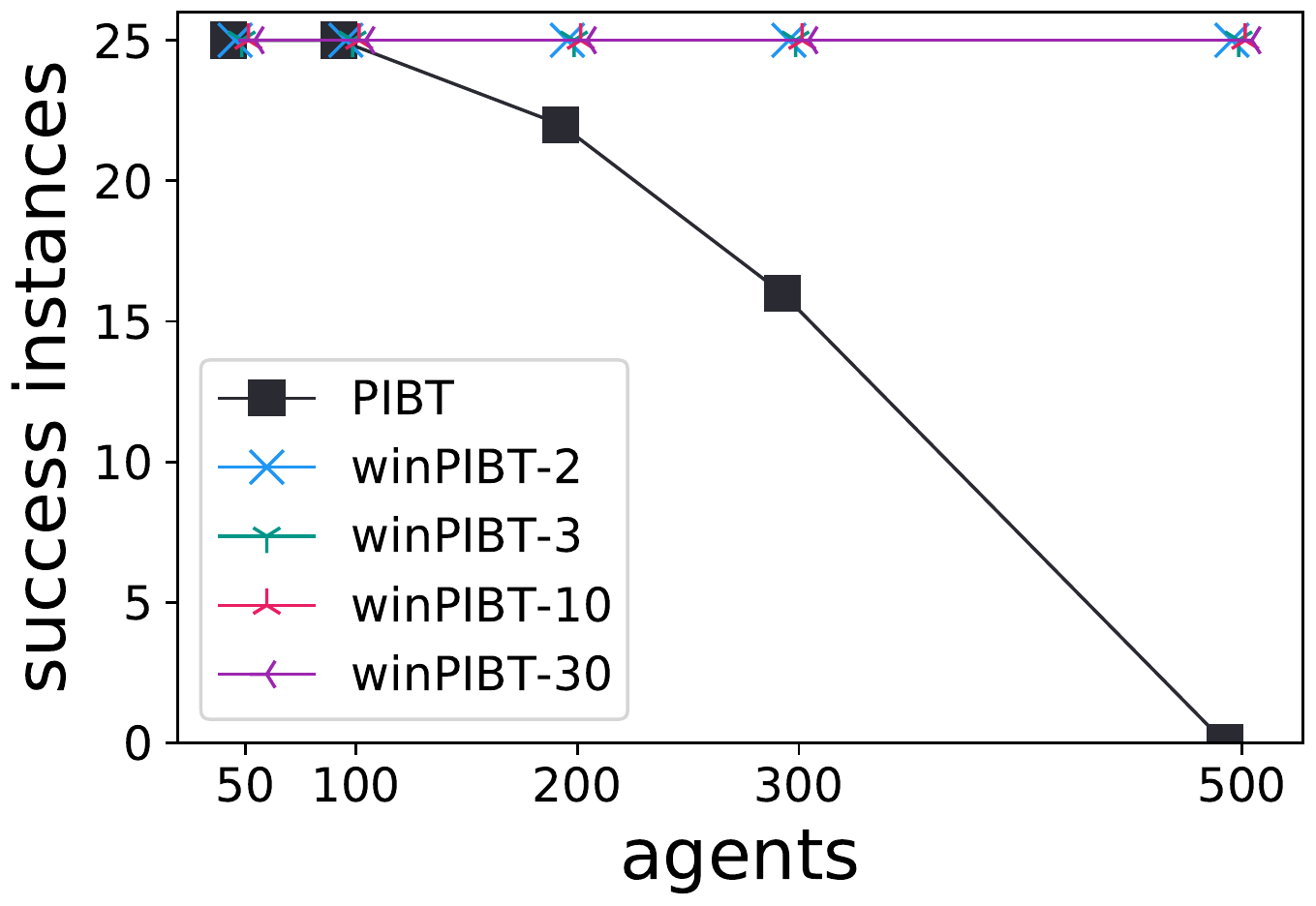}
  \end{minipage}
  \\
  \begin{minipage}{0.08\hsize}
   \centering
   {\scriptsize \ost}
   \vspace{0.1cm}
   \includegraphics[width=1\hsize]{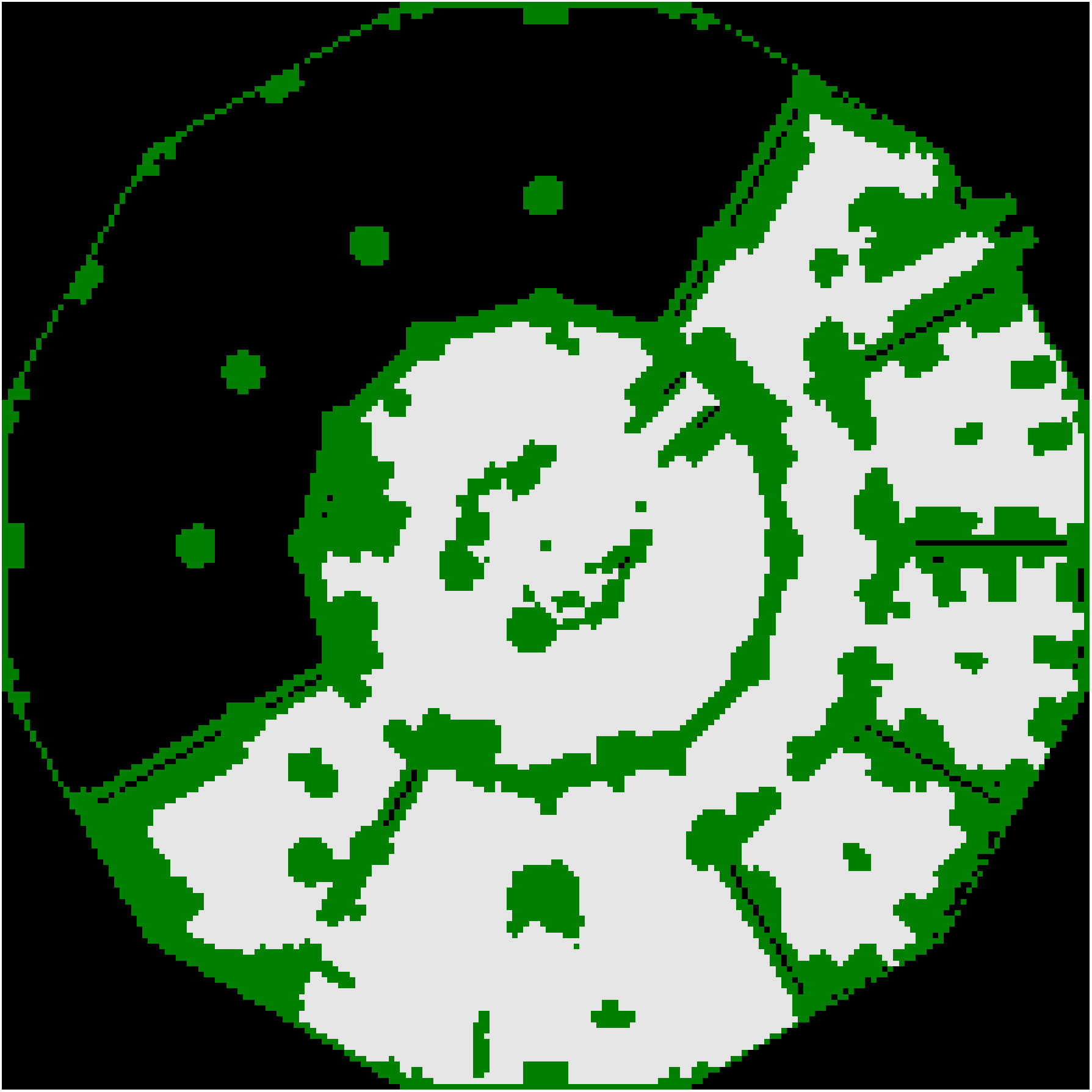}
  \end{minipage}
  \begin{minipage}{0.22\hsize}
   \centering
   \includegraphics[width=1\hsize]{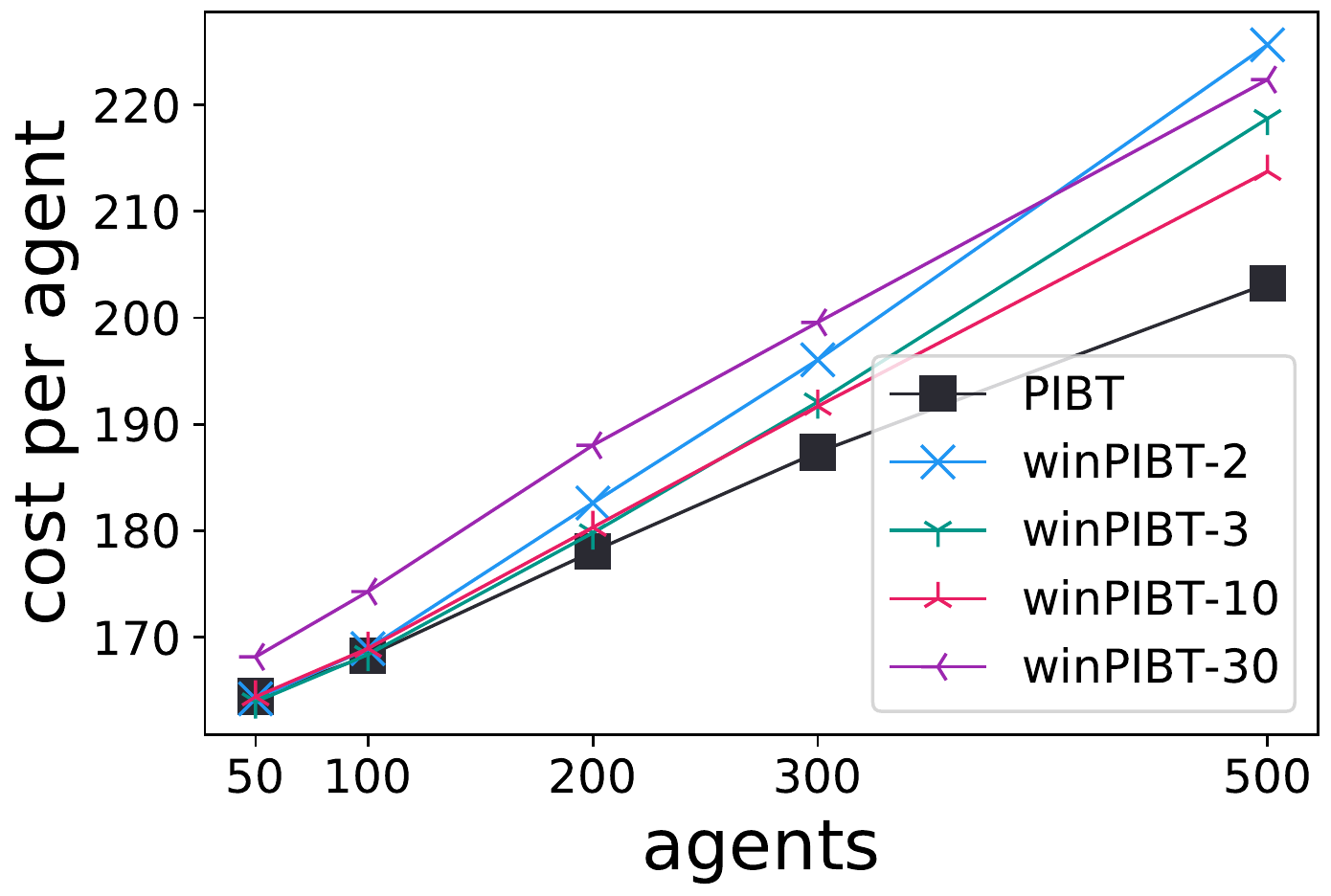}
  \end{minipage}
  \begin{minipage}{0.22\hsize}
   \centering
   \includegraphics[width=1\hsize]{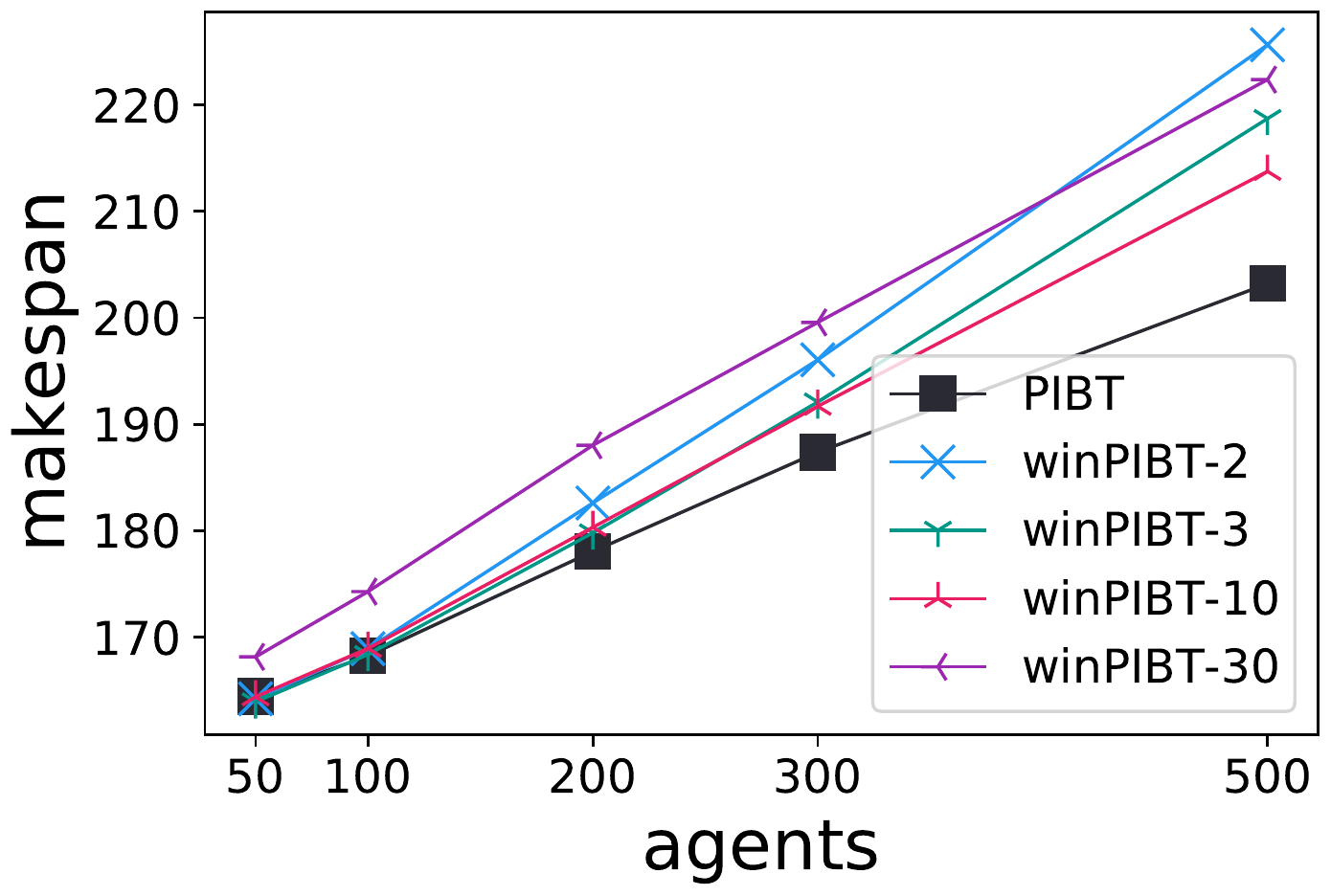}
  \end{minipage}
  \begin{minipage}{0.22\hsize}
   \centering
   \includegraphics[width=1\hsize]{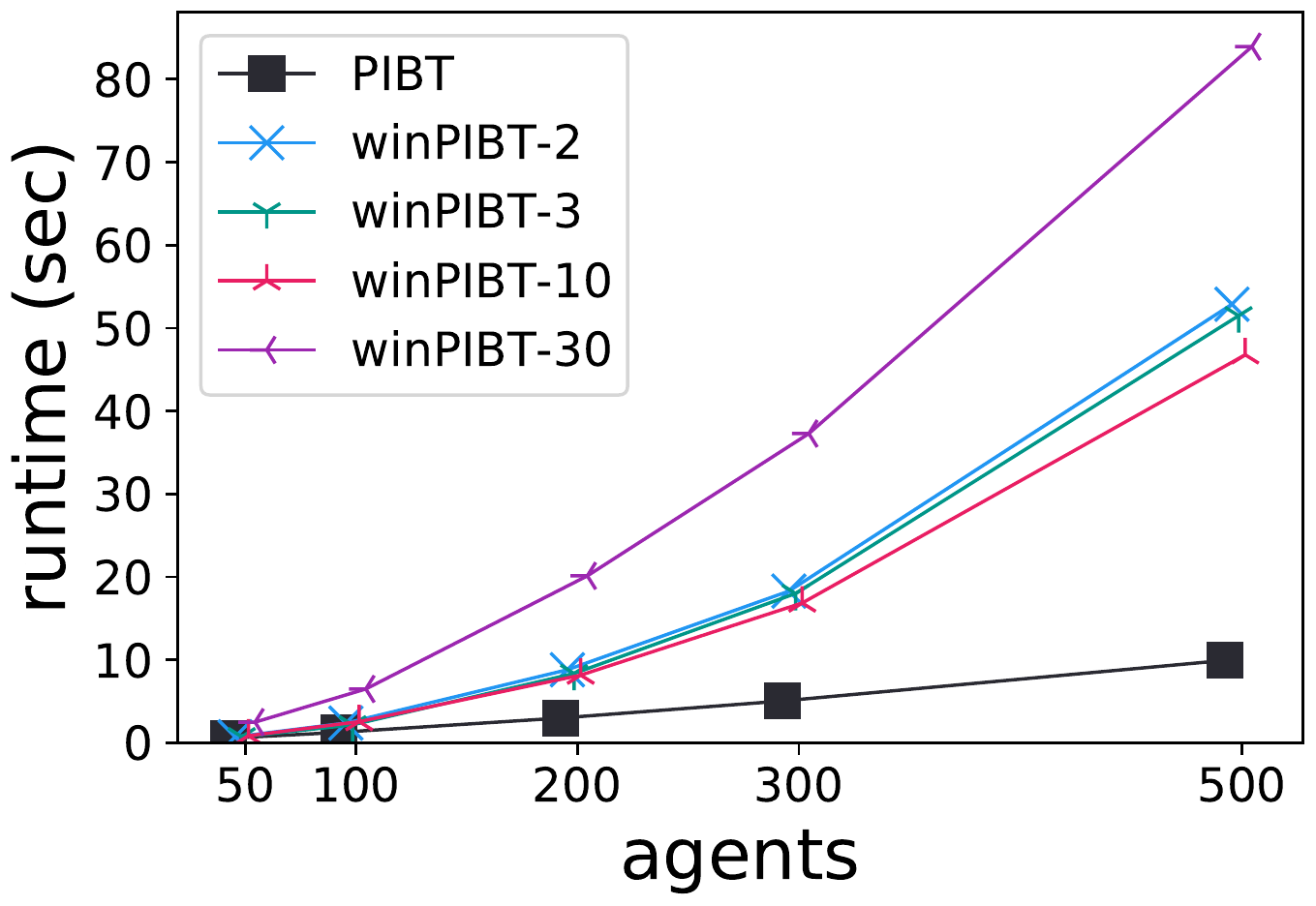}
  \end{minipage}
  \begin{minipage}{0.22\hsize}
   \centering
   \includegraphics[width=1\hsize]{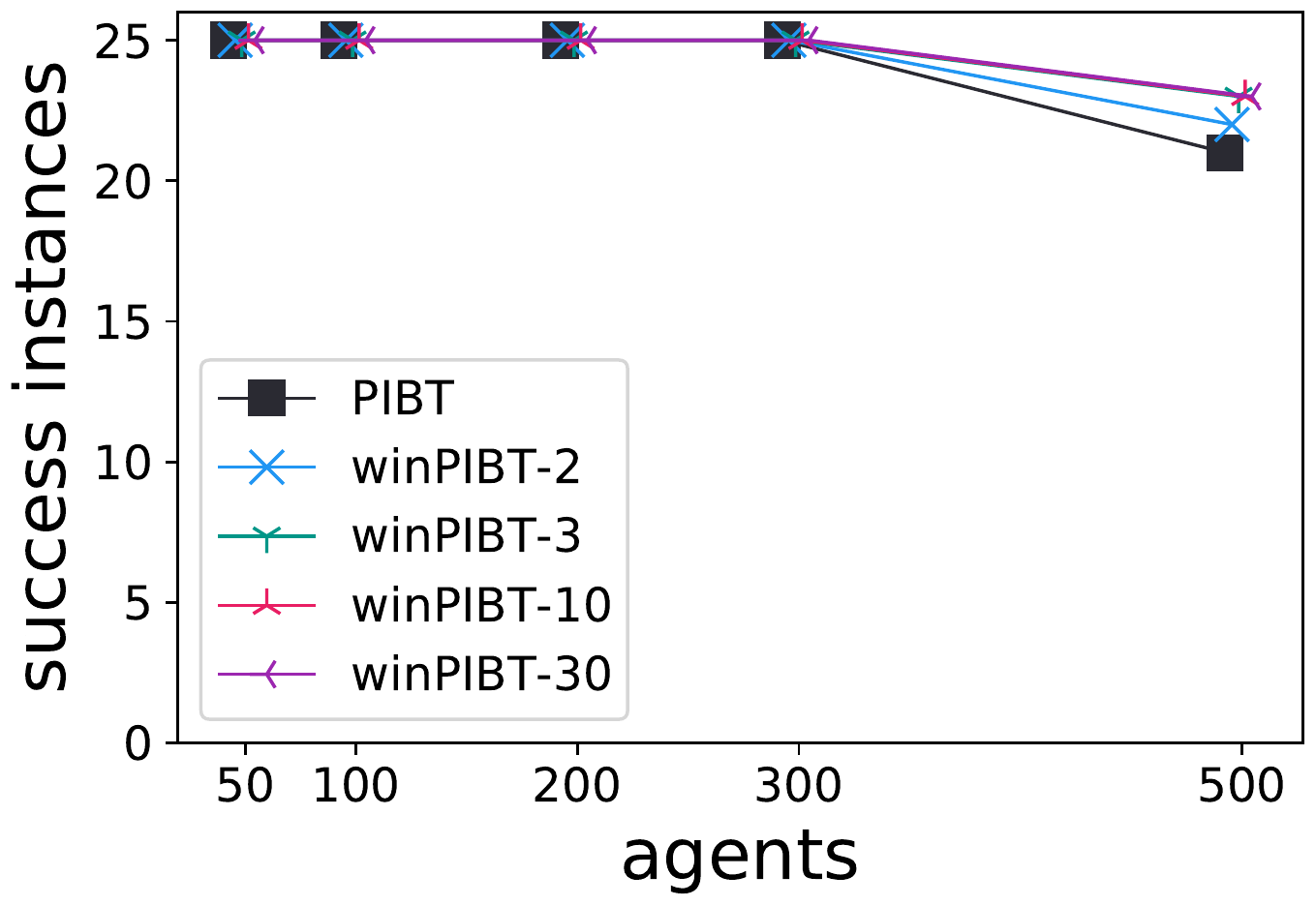}
  \end{minipage}
 \end{tabular}
 \caption{Results of classical MAPF with MAPF benchmark.
   Average scores over instances that were successfully solved by all solvers are shown.
   In the case with 500 agents in \emptymid, the scores are calculated excluding PIBT since PIBT failed in all instances.
   Note that, as for cost per agent and makespan, the vertical axes do not start from 0 to highlight the difference between solvers.
   }
 \label{fig:mapf-result-2}
\end{figure*}
\begin{figure*}[h!]
 \centering
 \begin{tabular}{cccc}
  \begin{minipage}{0.23\hsize}
   \centering
   {\scriptsize \kivalike}
   \vspace{0.1cm}
   \includegraphics[width=1\hsize]{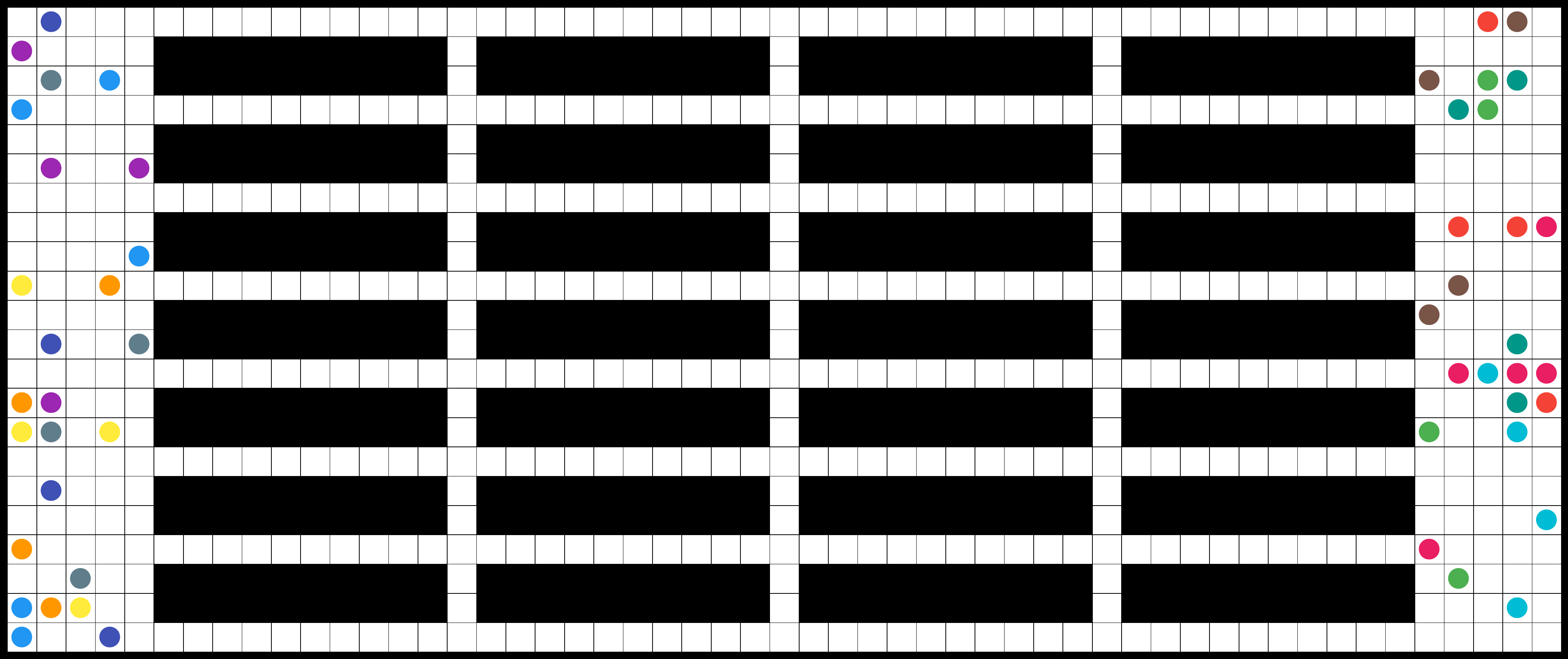}
  \end{minipage}
  \begin{minipage}{0.22\hsize}
   \centering
   \includegraphics[width=1\hsize]{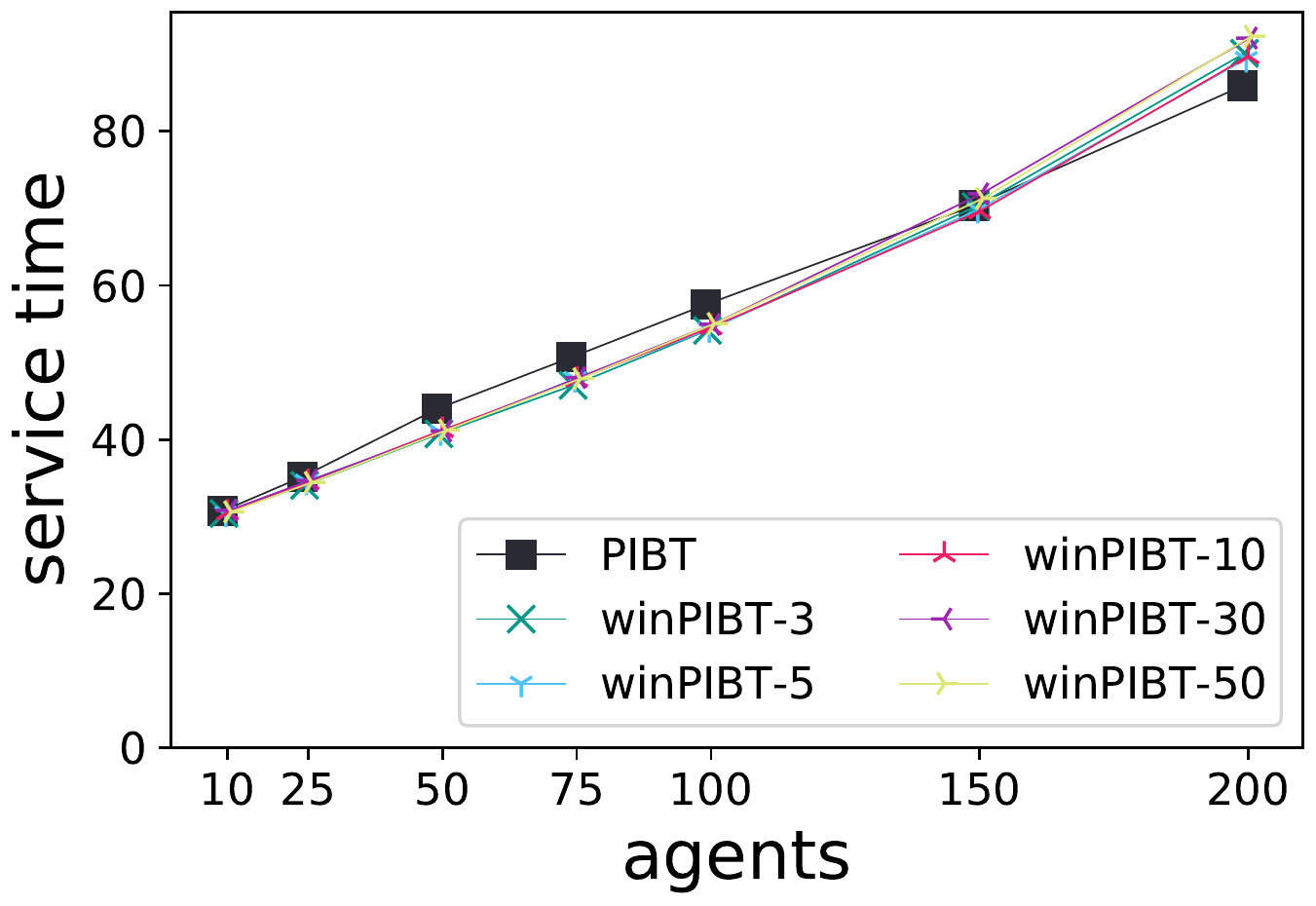}
  \end{minipage}
  \begin{minipage}{0.22\hsize}
   \centering
   \includegraphics[width=1\hsize]{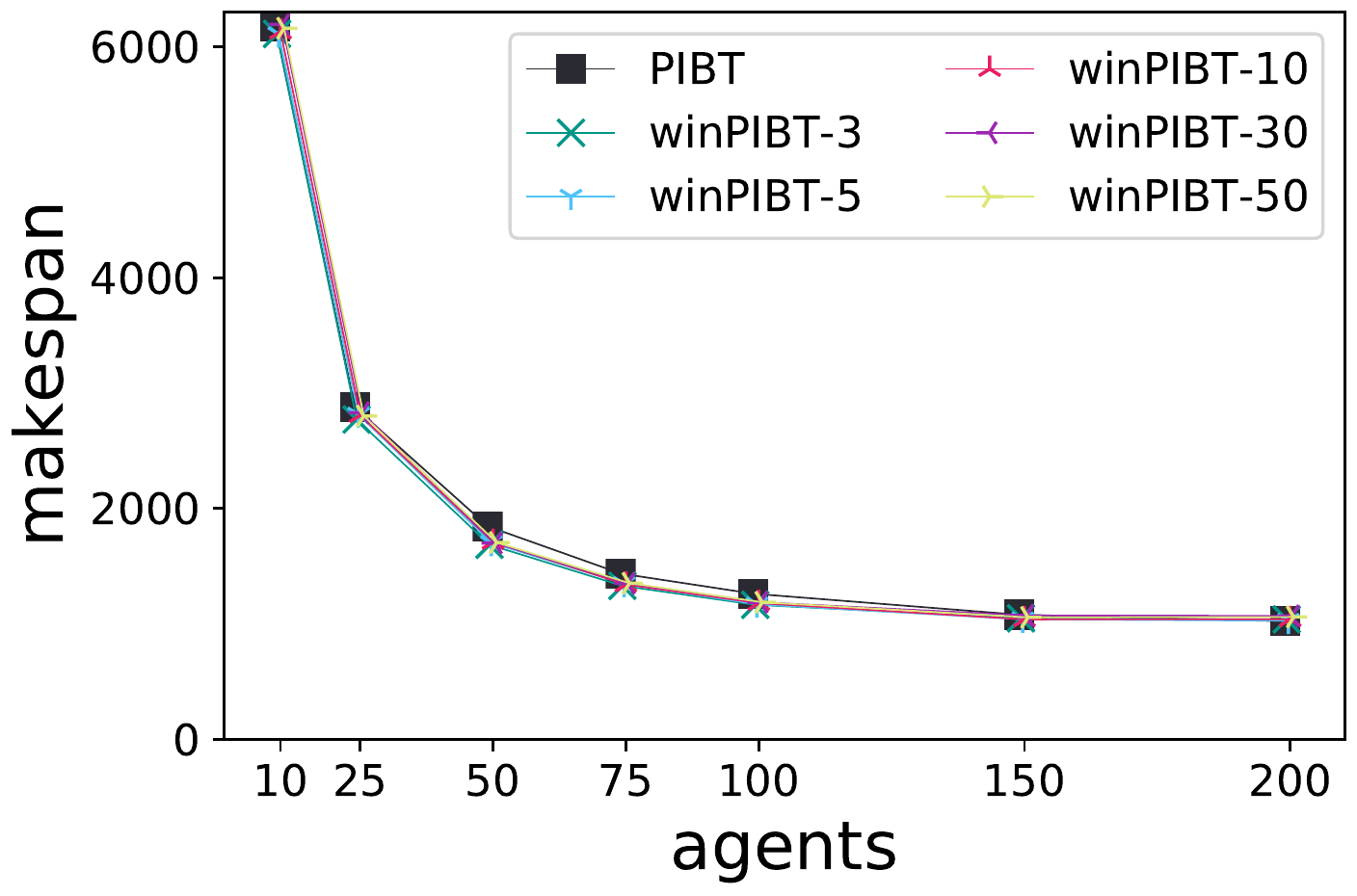}
  \end{minipage}
  \begin{minipage}{0.22\hsize}
   \centering
   \includegraphics[width=1\hsize]{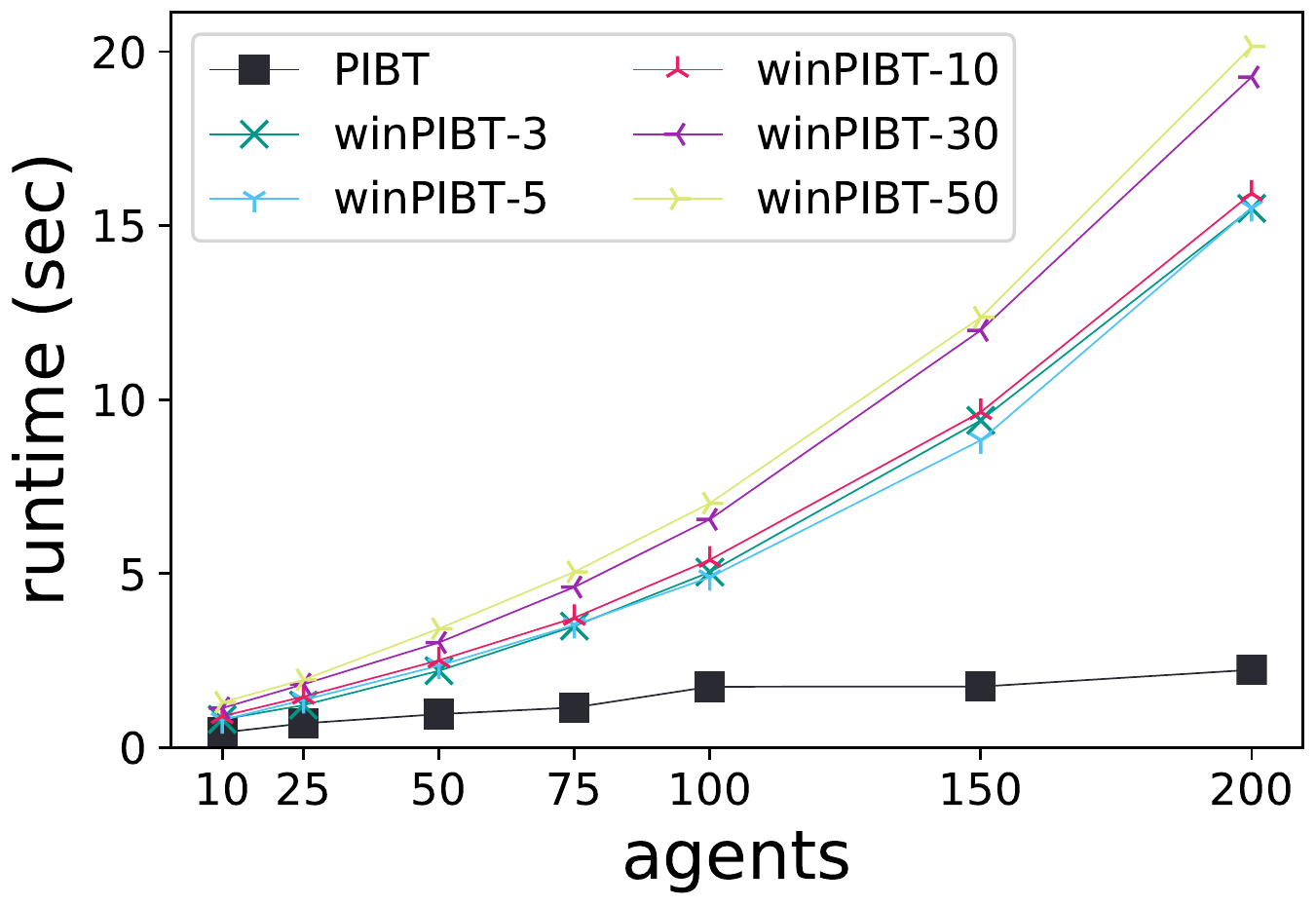}
  \end{minipage}
 \end{tabular}
 \caption{Results of \naive iterative MAPF in \kivalike. Average scores are shown.}
 \label{fig:imapf-result}
\end{figure*}

\subsubsection{Iterative Use}
\label{subsubsec:iterative}
For iterative use like Multi-agent Pickup and Delivery~\cite{ma2017lifelong}, it is meaningless to force agents to stay at their goal locations from their arrival until $t=\alpha$, i.e., $\alpha$ in function $\mathsf{winPIBT}$ can be treated more flexibly.
Once an agent reaches its destination, the agent can immediately return the backtracking by adding the following modifications in function $\mathsf{winPIBT}$ [Algorithm~\ref{algo:func-winpibt}].
Let $\delta$ be the timestep when an agent $a_{i}$ reaches its destination $g_{i}$ according to the calculated path and $\delta \leq \alpha$.
First, register the ideal paths until timestep $\delta$, not $\alpha$ [Line~\ref{algo:func-winpibt:register1},\ref{algo:func-winpibt:register2}].
Second, replace $\alpha$ with $\delta$ [Line~\ref{algo:func-winpibt:while-start}--\ref{algo:func-winpibt:endwhile}].
As a result, $a_{i}$ reserves its path until timestep $\delta$ and unnecessary reservations are avoided.

\subsubsection{Decentralized Implementation}
PIBT with decentralized fashion requires that each agent senses its surroundings to detect potential conflicts, then, communicates with others located within 2-hops, which is the minimum assumption to achieve conflict-free planning.
The part of priority inheritance and backtracking can be performed by information propagation.
\winpibt with decentralized fashion is an almost similar way to PIBT, however, it requires agents to sense and communicate with other agents located within $2{w}$-hops, where $w$ is the maximum window size that agents are able to take.
In this sense, there is an explicit trade-off; \emph{To do better anticipation, agents need expensive ability for sensing and communication}.

\section{Evaluation}
\label{sec:evaluation}
This section evaluates the performance of \winpibt quantitatively by simulation.
Our experiments are twofold: classical MAPF and \naive iterative MAPF.
The simulator was developed in C++
\footnote{
The code is available at \url{https://github.com/Kei18/pibt}
}, and all experiments were run on a laptop with Intel Core i5 1.6GHz CPU and 16GB RAM.
\astar was used to obtain the shortest paths satisfying constraints.

\subsection{Classical MAPF}
\subsubsection{Basic Benchmark}
To characterize the basic aspects of the effect of the window size, we first tested \winpibt in four carefully chosen fields, while fixing the number of agents.
Three fields (\fourthree, \bridge, \twobridge; Fig.~\ref{fig:mapf-result-1}) are original.
In these fields, 10 scenarios were randomly created such that starts and goals were set to nodes in left/right-edge and in right/left-edge, respectively.
The warehouse environment (\kivalike; Fig.~\ref{fig:imapf-result}) is from~\cite{cohen2015feasibility}.
In \kivalike, 25 scenarios were randomly created such that starts and goals were set in left/right-space, right/left-space, respectively.
As baselines for path efficiency, we obtained optimal and bounded sub-optimal solutions by Conflict-based Search (CBS)~\cite{sharon2015conflict} and Enhanced CBS (ECBS)~\cite{barer2014suboptimal}.
We also tested PIBT as a comparison.

We report the sum of cost (SOC) in Fig.~\ref{fig:mapf-result-1}.
We observe that no window size is dominant, e.g., $3$ in \fourthree, $6$ in \twobridge, $50$ in \kivalike work well respectively, and there is little effect of window size in \bridge.
Intuitively, efficient window size seems to depend on the length of narrow passages if detours exist, as shown in Fig~\ref{fig:motivating-example}.
In empty spaces, window size should be smaller to avoid unnecessary interference such as in Fig.~\ref{fig:winpibt-reservation:inconvenient}.
Although PIBT can be seen as \winpibt with window one, it may take time to reach the termination condition even in empty spaces due to the kind of livelock situations (see \fourthree).

\subsubsection{MAPF Benchmark}
Next, we tested \winpibt via MAPF benchmark~\cite{stern2019multi} while changing the number of agents.
Two maps (\emptymid and \ost) were chosen and 25 scenarios (random) were used.
Initial locations and destinations were given in order following each scenario, depending on the number of agents.
PIBT was also tested.
\winpibt or PIBT were considered failed when they could not reach termination conditions after 1000~timesteps.
These cases indicate occurrences of deadlock or livelock.
The former is due to the lack of graph condition, and the latter is due to dynamic priorities.

Fig.~\ref{fig:mapf-result-2} shows
1)~the results of cost per agent, i.e., normalized SOC,
2)~the makespan,
3)~the runtime, and
4)~the number of successful instances.
A nice characteristic of \winpibt is to mitigate livelock situations occurring in PIBT, regardless of the window size (see \emptymid).
The livelock in PIBT is due to oscillations of agents around their goals by the dynamic priorities.
\winpibt can improve this aspect with a longer lookahead.
As for cost, PIBT works better than \winpibt in tested cases since those two maps have no explicit detours like two-bridge.
Runtime results seemed to correlate with the window size, e.g., they took long time when the window size is 30.
Runtime results also correlate with the makespan since (win)PIBT solves problems in online fashion.
This is why the implementation with a small window took time compared with the larger one.

\subsection{\Naive Iterative MAPF}
We used \kivalike as testebed for \naive iterative MAPF.
The number of tasks $K$ for the termination was set to $2000$.
We tried $10$ repetitions of each experiment with randomly set initial positions.
New goals were given randomly.
\winpibt is modified for iterative use.
PIBT was used as comparison.
Note that since \kivalike is \graphcond, PIBT and \winpibt are ensured to terminate.

The results of 1) service time, 2) makespan and 3) runtime are shown in Fig.~\ref{fig:imapf-result}.
The effect of window size on path efficiency is marginal;
Depending on the number of agents, there is little improvement compared with PIBT.
We estimate that the reason for the small effect is as follows.
First, both small and large window have bad situations.
In the lifelong setting as used here, agents may encounter both good and bad situations.
Second, truly used window size becomes smaller than the parameter, since the algorithm used here does not allow for agents with lower priorities to disturb planning by higher priorities.
This characteristic combined with the treatment for iterative use may counteract the effect of the window size.

\subsection{Discussion}
In generally, in long aisles where agents cannot pass each other, PIBT plans awkward paths as explained in Fig.~\ref{fig:motivating-example}.
\winpibt can improve PIBT in this aspect (see \kivalike in classical MAPF), however, the empirical results demonstrate the limitation of the fixed window.
Fortunately, \winpibt allows agents to take different window sizes, meaning that, agents can adjust their window adaptively depending on situations, e.g., their locations, the density of agents.
For instance, it seems to be effective to set window size which is enough to cover whole of aisles when an agent tries to enter such aisles.
This kind of flexible solution is expected not only to improve path efficiency but also to reduce computation time.
Clarifying the relationship between window size and path efficiency helps to develop an adaptive version of \winpibt.
We believe that this direction will provide a powerful solution for iterative MAPF.

\section{Conclusion}
\label{sec:conclusion}
This paper introduces \winpibt which generalizes PIBT regarding the time window.
We define a \safe condition on all paths with different lengths and \winpibt relies on this concept.
The algorithm ensures the reachability for iterative MAPF in adequate properties of graphs, e.g., biconnected.
Empirical results demonstrate the potential of \winpibt by adjusting the window size.

Future work are as follows.
1)~Develop \winpibt with adaptive windows.
2)~Relax constraints on individual path planning, in such trivial cases as shown in Fig.~\ref{fig:winpibt-reservation:inconvenient}.


\bibliographystyle{ACM-Reference-Format}  
\bibliography{ref}  

\end{document}